\documentclass[11pt]{article}

\usepackage[bibliosources=refs.bib]{pomegranate}
\usepackage{xurl}

\DeclareOperator{\per}
\newcommand{\capop}{\operatorname{Cap}}
\DeclareOperator{\dist}
\DeclareOperator{\diag}
\DeclareOperator{\conv}
\newcommand{\eps}{\varepsilon}
\newcommand{\cB}{\mathcal B}
\newcommand{\cD}{\mathcal D}
\newcommand{\cE}{\mathcal E}
\newcommand{\cG}{\mathcal G}
\newcommand{\KL}{D_{\mathrm{KL}}}
\newcommand{\Bethe}{\operatorname{Bethe}}
\newcommand{\Htwo}{H_2}

\hypersetup{
  pdftitle={Beyond the Bethe Approximation of the Permanent},
  pdfauthor={Nima Anari}
}

\title{Beyond the Bethe Approximation of the Permanent}
\author[1]{Nima Anari\thanks{The accompanying Lean 4 formalization is available at
\url{https://github.com/nimaanari/formalization-beyond-bethe}.  The corresponding
Palomar Registry entry is available at
\url{https://palomar-registry.org/entry?id=PALOMAR-2026-08-29-000006&version=1}.}}
\affil[1]{Stanford University, \texttt{anari@stanford.edu}}
\date{}

\begin{document}

\maketitle

\begin{abstract}
The canonical Bethe approximation gives a deterministic approximation to the
permanent of every nonnegative matrix within a factor of \((\sqrt2)^n\).  We
improve the base of this exponential factor: for some absolute constant
\(c<\sqrt2\), there is a deterministic polynomial-time
\(c^n\)-approximation for the permanent of every nonnegative matrix.  This shows
that the canonical Bethe guarantee is not a barrier for deterministic
approximation of the permanent.  The proof augments the Bethe lower bound with
a new certificate tailored to matrices on which that lower bound loses nearly
the full factor.

The author supplied the high-level plan of attack, and the proof was developed
in an interaction with ChatGPT 5.6 Sol Pro.  The author subsequently verified
the results.  Codex assisted with proof checking, manuscript assembly, and
typesetting.
\end{abstract}

\section{Introduction}

The permanent of a matrix \(A\in\R_{\ge0}^{n\times n}\) is
\[
  \per(A)=\sum_{\sigma\in S_n}\prod_{i=1}^n A_{i,\sigma(i)}.
\]
Here \(S_n\) is the set of permutations of \(\set{1,\ldots,n}\).
Computing or approximating the permanent is one of the canonical problems in
theoretical computer science and the field of counting and sampling.  Valiant
proved that computing the permanent is \(\Class{\#P}\)-hard even for \(0/1\)
matrices
\cite{Valiant1979}.  On the other hand, Jerrum, Sinclair, and Vigoda gave a
randomized fully polynomial-time approximation scheme for every nonnegative
matrix \cite{JerrumSinclairVigoda2004}.  The best general deterministic
algorithms remain exponentially worse.

A particularly natural deterministic approximation is the \emph{Bethe
permanent}.  It is the maximum of a concave function over the Birkhoff
polytope, the set of doubly stochastic matrices, and can therefore be computed
efficiently; we denote it by \(\Bethe(A)\).  Gurvits proved that it
is a lower bound for the permanent, building on Schrijver's permanental
inequality \cite{Schrijver1998,Gurvits2011}; Vontobel developed its variational
and graph-cover interpretations \cite{Vontobel2013}.  Anari and Rezaei proved
the sharp universal sandwich
\[
  \Bethe(A)\le \per(A)\le 2^{n/2}\Bethe(A)
\]
\cite{AnariRezaei2024}.  The factor \(2^{n/2}\) is attained by a block diagonal
matrix with \(2\times2\) all-ones blocks.  Thus no analysis of the Bethe
permanent alone can improve the base \(\sqrt2\).

This paper goes beyond that barrier by augmenting the Bethe certificate rather
than replacing it.  Our main result is the following.

\begin{theorem}[Main theorem]\label{thm:main}
There is an absolute constant \(c<\sqrt2\) and a
deterministic polynomial-time algorithm that, given
\(A\in\Q_{\ge0}^{n\times n}\), returns a number \(L(A)\) satisfying
\[
  L(A)\le \per(A)\le c^n L(A).
\]
\end{theorem}

We make no attempt to optimize \(c\).  The point is the uniform improvement in
the base of the exponential approximation; optimizing the constants would
obscure the argument.

\paragraph*{Independent work.}
Independently, Narang and Perkins proved a similar result: a
deterministic polynomial-time \((\sqrt2-\varepsilon)^n\)-approximation for some
absolute constant \(\varepsilon>0\) \cite{NarangPerkins2026}.  Their approach
also exploits nearly isolated weighted \(2\times2\) blocks, through a structural
correction to the Bethe approximation.

\subsection{Proof overview}

The proof couples a new lower-bound certificate with a stability analysis of
the Bethe upper bound.  For this overview, suppose that \(A\) is positive and
let
\[
  \Delta(A)=\frac n2\log2+\log\Bethe(A)-\log\per(A)
\]
be the slack in that upper bound.  If \(\Delta(A)=\Omega(n)\), then the Bethe
permanent already beats \((\sqrt2)^n\) by an exponential factor.  We therefore
only need a better certificate when \(\Delta(A)\) is a small constant multiple
of \(n\).  The argument has three main ideas.

\paragraph*{A lower bound that remembers pairs of rows.}
A matching partitions the rows into pairs and unmatched singletons.  A
singleton contributes a linear polynomial whose capacity is its contribution
to the Bethe lower bound.  For a pair of rows \(r,s\), we instead use the
quadratic polynomial
\[
  Q_{rs}(z)=\sum_{j<k}
  \left(A_{rj}A_{sk}+A_{rk}A_{sj}\right)z_jz_k.
\]
The coefficient of \(z_jz_k\) is the total weight of the two ways to assign
rows \(r,s\) to the distinct columns \(j,k\).  We multiply the polynomials of
all row clusters and pair the result with another polynomial that enforces
that every column is used exactly once.  A coefficient inequality for real
stable polynomials---a class defined in \cref{sec:preliminaries}---then turns
this encoding into a lower bound for \(\per(A)\)
\cite{AnariOveisGharan2017}.  On an exact \(2\times2\) all-ones block, the
paired cluster recovers the factor \(2\) lost by the Bethe permanent.

\paragraph*{Near equality forces approximate \(2\times2\) blocks.}
Sample a permutation \(\sigma\) with probability proportional to its weight
\(\prod_iA_{i,\sigma(i)}\), and let
\(P_{ij}=\P_{\sigma}{\sigma(i)=j}\).  Thus row \(i\) of \(P\) records the
marginal distribution of the column assigned to row \(i\).  Anari and Rezaei
analyze the following sequential experiment.  Order the rows uniformly at
random and process them in that order; when row \(i\) is reached, choose an
unused column with probability proportional to the corresponding entries of
\(P_i\).  This experiment produces another distribution on permutations.
Their upper bound compares it with the original weighted distribution and
then applies a one-row inequality.

We keep the information lost in both steps.  This gives an exact decomposition
of \(\Delta(A)\) into three nonnegative contributions: the suboptimality of
\(P\) in the Bethe program, the sum of the deficits in the one-row inequality,
and the average relative entropy between the true and sequential
distributions.  Consequently, when \(\Delta(A)\) is small, all three
contributions are small.

The one-row inequality is stable: a row with small deficit must be close to a
probability vector supported uniformly on two columns.  Hence most rows of
\(P\) have two entries near \(1/2\).  These large entries form a bipartite graph
of maximum degree two.  The relative-entropy contribution rules out many long
cycles in this graph, so most rows lie in approximate copies of \(K_{2,2}\),
the complete bipartite graph on two rows and two columns.  This is the rigidity
statement behind the improvement.

\paragraph*{Transfer the blocks and collect their gain.}
The marginal matrix \(P\) is not available to the algorithm.  We optimize a
slightly regularized Bethe objective instead, adding a small multiple of the
sum of the Shannon entropies of its rows.  The unique optimizer \(X\) is
positive.  Its Karush--Kuhn--Tucker optimality equations give positive row and
column scaling factors \(r_i,c_j\) such that
\[
  A_{ij}=r_i c_j X_{ij}^{1+\tau}(1-X_{ij}).
\]
Here \(\tau\) is a small positive regularization parameter.  A global identity
derived from this factorization measures how much the relevant coordinates
can change between \(P\) and \(X\).  It shows that most of the approximate
\(K_{2,2}\)'s found in \(P\) remain useful at \(X\).  On each such block, the
paired certificate improves on the two singleton certificates by a constant
factor.  A maximum-weight row matching collects these gains on a linear number
of disjoint blocks, improving the lower bound by \(e^{\Omega(n)}\).

The key point is the interaction between the two sides of the argument.  The
stability analysis identifies precisely where the Bethe bound can be nearly
tight, while the paired certificate is designed to recover weight on those
same configurations.

\subsection{Related work}

The Bethe permanent was introduced and analyzed systematically by Vontobel
\cite{Vontobel2013}.  Its lower-bound property follows from Schrijver's
inequality and Gurvits's variational formulation
\cite{Schrijver1998,Gurvits2011}.  Anari and Rezaei proved the matching
\(2^{n/2}\) upper bound \cite{AnariRezaei2024}.  Our proof begins with their
sequential-distribution argument but uses its discarded relative-entropy term
as a structural statistic.

Graph covers give another interpretation of the Bethe permanent.  Ng and
Vontobel studied the role of double covers \cite{NgVontobel2022}; Huang,
Kashyap, and Vontobel developed degree-\(M\) Bethe and Sinkhorn permanent
bounds \cite{HuangKashyapVontobel2024}.  These works analyze the Bethe
approximation and its finite-cover counterparts.  Our certificate instead
recovers weight that the Bethe objective necessarily loses on
\(2\times2\) blocks.

The stable-polynomial input is the multiaffine coefficient inequality of Anari
and Oveis Gharan \cite{AnariOveisGharan2017}.  It proves one-sidedness of each
paired certificate; the global argument deciding which pairs are profitable
is information-theoretic.

\paragraph*{Organization.}
\Cref{sec:preliminaries} records notation, the Bethe bound, and the two
stable-polynomial facts used in the proof.  \Cref{sec:paired} constructs the
paired lower certificate.  \Cref{sec:slack,sec:row-stability} prove the exact
slack identity and the stability of its row term.  \Cref{sec:cycles} extracts
the robust \(K_{2,2}\) structure.  \Cref{sec:transfer,sec:gain} transfer this
structure to the regularized optimizer and prove the paired gain.
\Cref{sec:completion} combines the estimates.  \Cref{sec:zeros} treats zero
entries and finite-precision computation.

\section{Preliminaries}\label{sec:preliminaries}

\paragraph*{Notation.}
We write \([n]=\set{1,\ldots,n}\) and let \(S_n\) denote the set of
permutations of \([n]\).  The Birkhoff polytope is
\[
  \cB_n
  =\set*{X\in\R_{\ge0}^{n\times n}
    \given \sum_jX_{ij}=1\ \forall i,
    \ \sum_iX_{ij}=1\ \forall j}.
\]
For a probability vector \(p\), its Shannon entropy is
\[
  H(p)=-\sum_jp_j\log p_j,
\]
with the convention \(0\log0=0\).  We write
\(\Htwo(t)=-t\log t-(1-t)\log(1-t)\) for binary entropy.  For probability
measures \(\mu,\nu\) on the same finite set, their Kullback--Leibler (KL)
divergence is
\[
  \KL(\mu\Vert\nu)
  =\sum_x\mu(x)\log\frac{\mu(x)}{\nu(x)}.
\]
All logarithms are natural.  We use the same formula to write \(H(\mu)\) for
the entropy of a finite probability measure, \(H(Z)\) for the entropy of a
finite-valued random variable, and \(H(Z\mid W)\) for conditional entropy.

\subsection{The Bethe permanent}

For \(A\in\R_{>0}^{n\times n}\) and \(X\in\cB_n\), define the Bethe objective
\begin{equation}\label{eq:bethe-objective}
  \beta_A(X)
  =\sum_{i,j}
  \left[
    X_{ij}\log\frac{A_{ij}}{X_{ij}}
    +(1-X_{ij})\log(1-X_{ij})
  \right].
\end{equation}
At \(X_{ij}\in\set{0,1}\), we interpret the summands by continuity.  The Bethe
permanent is
\[
  \Bethe(A)=\exp\left(\max_{X\in\cB_n}\beta_A(X)\right).
\]
For a nonnegative matrix \(A\), we use the same definition with the convention
that \(X_{ij}\log(A_{ij}/X_{ij})=-\infty\) when \(A_{ij}=0<X_{ij}\), and
that this term is zero when \(X_{ij}=0\).  Equivalently, this is the limit of
\(\Bethe(A+\delta J)\) as \(\delta\downarrow0\), where \(J\) is the all-ones
matrix.
Vontobel proved that the objective is concave on \(\cB_n\)
\cite{Vontobel2013}.  The following sharp sandwich combines the lower bound of
Gurvits with the upper bound of Anari and Rezaei
\cite{Gurvits2011,AnariRezaei2024}.

\begin{theorem}[Bethe sandwich]\label{thm:bethe-sandwich}
For every \(A\in\R_{\ge0}^{n\times n}\),
\begin{equation}\label{eq:bethe-sandwich}
  \Bethe(A)\le\per(A)\le2^{n/2}\Bethe(A).
\end{equation}
Both inequalities are tight.
\end{theorem}

We first prove the result for positive matrices.  Zero entries are treated
by a quantitative smoothing argument in \cref{sec:zeros}.

\subsection{Capacity and real stability}

We need one coefficient inequality for multivariate polynomials.  Its
variational quantity is the following scale-invariant lower envelope.
Let \(q\in\R_{\ge0}[z_1,\ldots,z_m]\) be homogeneous of degree \(d\).  For
\(\alpha\in\R_{\ge0}^m\) with \(\sum_j\alpha_j=d\), define its capacity at
\(\alpha\) by
\begin{equation}\label{eq:capacity}
  \capop_\alpha(q)
  =\inf_{z\in\R_{>0}^m}\frac{q(z)}{z^\alpha}.
\end{equation}
Here and below \(z^\alpha=\prod_jz_j^{\alpha_j}\).

A polynomial is \emph{multiaffine} if every variable has degree at most one.
A polynomial \(q\in\R[z_1,\ldots,z_m]\) is \emph{real stable} if
\(q(z_1,\ldots,z_m)\ne0\) whenever every \(z_j\) has positive imaginary part.
Products of real stable polynomials in disjoint variable sets are real stable.
We use the following two facts.

\begin{lemma}[Quadratic stability criterion]\label{lem:quadratic-stability}
Let \(q\) be a nonzero homogeneous multiaffine quadratic with nonnegative
coefficients.  If the Hessian of \(q\) has at most one positive eigenvalue,
then \(q\) is real stable.
\end{lemma}

\begin{proof}
Write \(q(z)=\frac12z^\top Hz\), where \(H\) is the Hessian.  Suppose that
\(q(x+iy)=0\) for some \(x\in\R^m\) and \(y\in\R_{>0}^m\).  Separating real
and imaginary parts gives
\[
  x^\top Hx=y^\top Hy,
  \qquad
  x^\top Hy=0.
\]
The nonnegative coefficients and \(y>0\) imply \(y^\top Hy=2q(y)>0\).  Since
\(H\) has at most one positive eigenvalue, its quadratic form is
nonpositive on the \(H\)-orthogonal complement of \(y\), namely on the
vectors \(x\) satisfying \(x^\top Hy=0\).  Thus
\(x^\top Hx\le0\), contradicting the first equality.  This is the
degree-two case of Br\"and\'en's half-plane-property criterion
\cite{Branden2007}.
\end{proof}

For a set \(S\subseteq[m]\), write \(z^S=\prod_{j\in S}z_j\).  For
multiaffine polynomials
\(p(z)=\sum_Sp_Sz^S\) and \(q(z)=\sum_Sq_Sz^S\), write
\[
  \langle p,q\rangle=\sum_Sp_Sq_S.
\]
Thus this inner product pairs coefficients of the same squarefree monomial.

\begin{theorem}[Stable coefficient inequality]\label{thm:stable-pairing}
Let \(p,q\in\R_{\ge0}[z_1,\ldots,z_m]\) be homogeneous, multiaffine, real
stable polynomials of the same degree \(d\).  For every
\(\alpha\in[0,1]^m\) with \(\sum_j\alpha_j=d\),
\begin{equation}\label{eq:stable-pairing}
  \langle p,q\rangle
  \ge
  \prod_{j=1}^m
    \alpha_j^{\alpha_j}(1-\alpha_j)^{1-\alpha_j}
  \capop_\alpha(p)\capop_\alpha(q).
\end{equation}
The conventions \(0^0=1\) and \(1^0=1\) are used at the boundary.
\end{theorem}

This is the multiaffine specialization of the coefficient inequality of Anari
and Oveis Gharan \cite{AnariOveisGharan2017}.  The inner product in
\cref{eq:stable-pairing} pairs the same squarefree monomial in \(p\) and \(q\),
rather than complementary monomials.

We use the following elementary entropy certificate for capacity.  Its
one-sided form is exactly what the proof and the algorithm need.

\begin{lemma}[Entropy certificate for capacity]\label{lem:capacity-certificate}
Let
\[
  q(z)=\sum_{E\in\cE}c_Ez^E,
  \qquad c_E>0,
\]
where \(\cE\subseteq\Z_{\ge0}^m\) and \(z^E=\prod_jz_j^{E_j}\).  If
\(\theta\) is a probability distribution on \(\cE\) with
\(\sum_E\theta_EE=\alpha\), then
\begin{equation}\label{eq:capacity-certificate}
  \sum_{E\in\cE}\theta_E\log\frac{c_E}{\theta_E}
  \le \log\capop_\alpha(q).
\end{equation}
Terms with \(\theta_E=0\) are interpreted by continuity.
\end{lemma}

\begin{proof}
Fix a feasible \(\theta\).  For every \(z>0\), the log-sum inequality gives
\begin{align*}
  \log q(z)-\langle\alpha,\log z\rangle
  &\ge
  \sum_E\theta_E
  \log\frac{c_Ez^E}{\theta_E}
  -\sum_E\theta_E\langle E,\log z\rangle\\
  &=\sum_E\theta_E\log\frac{c_E}{\theta_E}.
\end{align*}
Taking the infimum over \(z\) proves the claim.
\end{proof}

\section{A paired lower certificate}\label{sec:paired}

Throughout the structural argument we assume \(n\ge2\); the case \(n=1\) is
exact.
We augment the Bethe lower bound by clustering disjoint pairs of rows.  The
construction is valid for every interior point of the
Birkhoff polytope; later we evaluate it at the optimizer of a regularized
Bethe objective.

Fix a positive matrix \(A\in\R_{>0}^{n\times n}\) and an interior point
\(X\in\cB_n\), meaning that every entry of \(X\) is positive.  For a row
\(i\), define
\begin{equation}\label{eq:singleton}
  S_i(A,X)
  =\prod_j
    \left(\frac{A_{ij}}{X_{ij}}\right)^{X_{ij}}
    (1-X_{ij})^{1-X_{ij}}.
\end{equation}
The product of the singleton factors is exactly the exponential of the Bethe
objective:
\begin{equation}\label{eq:singleton-product}
  \prod_iS_i(A,X)=e^{\beta_A(X)}.
\end{equation}

For two distinct rows \(r,s\), let
\begin{equation}\label{eq:pair-polynomial}
  Q_{rs}(z)
  =\sum_{j<k}
    \left(A_{rj}A_{sk}+A_{rk}A_{sj}\right)z_jz_k,
  \qquad
  \alpha^{rs}_j=X_{rj}+X_{sj}.
\end{equation}
The coefficient of \(z_jz_k\) is the combined weight of the two assignments
of rows \(r,s\) to columns \(j,k\), while \(\alpha^{rs}_j\) is the total mass
that these two rows place on column \(j\).
The column constraints on \(X\) imply
\[
  0<\alpha^{rs}_j\le1,
  \qquad
  \sum_j\alpha^{rs}_j=2.
\]
Writing \(e_j\) for the \(j\)-th standard basis vector, the vectors in
\([0,1]^n\) whose coordinates sum to two form the convex hull of
\(e_j+e_k\), \(j<k\).  Thus \(\alpha^{rs}\) lies in the convex hull of the
exponent vectors of the monomials in \(Q_{rs}\), as required for its capacity.
Define
\begin{equation}\label{eq:pair-certificate}
  S^{\mathrm{pair}}_{rs}(A,X)
  =\prod_j(1-\alpha^{rs}_j)^{1-\alpha^{rs}_j}
    \capop_{\alpha^{rs}}(Q_{rs}),
  \qquad
  \Gamma_{rs}(A,X)
  =\frac{S^{\mathrm{pair}}_{rs}(A,X)}{S_r(A,X)S_s(A,X)}.
\end{equation}
If \(\alpha^{rs}_j=1\), the corresponding factor in the first product is
interpreted as \(0^0=1\).
The ratio \(\Gamma_{rs}\) measures the gain from treating rows \(r,s\) as one
pair rather than as two singleton rows.

\begin{theorem}[Paired lower certificate]\label{thm:paired-certificate}
For every matching \(M\) on the row set---that is, every collection of
disjoint unordered pairs of rows---
\begin{equation}\label{eq:paired-lower-bound}
  L_M(A,X)
  :=e^{\beta_A(X)}
    \prod_{\set{r,s}\in M}\Gamma_{rs}(A,X)
  \le\per(A).
\end{equation}
\end{theorem}

The theorem has two important features.  It is one-sided for every matching,
so the pairs may be chosen algorithmically, and it reduces the rest of the
paper to finding many pairs with \(\Gamma_{rs}>1\).  The proof encodes the row
clusters and the requirement of using every column once by two stable
polynomials whose coefficient inner product is the permanent.

\begin{proof}
\medskip\noindent\emph{The cluster polynomial.}
Let \(\cG\) be the partition of the row set whose nontrivial parts are the
pairs in \(M\).  Give each cluster \(C\in\cG\) its own copy
\(z_{C1},\ldots,z_{Cn}\) of the column variables.  For a singleton
\(C=\set{i}\), put
\[
  p_C(z_C)=\sum_jA_{ij}z_{Cj}.
\]
For a pair \(C=\set{r,s}\), put \(p_C(z_C)=Q_{rs}(z_C)\), and define
\begin{equation}\label{eq:cluster-product}
  p(z)=\prod_{C\in\cG}p_C(z_C).
\end{equation}

\medskip\noindent\emph{Stability.}
The singleton polynomials are positive linear
forms and are therefore real stable.  For a pair, write \(u=A_r\) and
\(v=A_s\), viewed as row vectors, and let \(u\circ v\) denote their
coordinatewise product.  The Hessian of \(Q_{rs}\) is
\begin{equation}\label{eq:pair-hessian}
  \nabla^2Q_{rs}
  =uv^\top+vu^\top-2\diag(u\circ v).
\end{equation}
Multiplying this matrix on the left and right by
\(\diag(u\circ v)^{-1/2}\) transforms it into
\[
  ab^\top+ba^\top-2I,
  \qquad
  a_j=\sqrt{u_j/v_j},
  \quad
  b_j=\sqrt{v_j/u_j}.
\]
Here \(I\) is the identity matrix.
The two possibly nonzero eigenvalues of \(ab^\top+ba^\top\) are
\[
  a^\top b+\norm{a}\norm{b}
  \qquad\text{and}\qquad
  a^\top b-\norm{a}\norm{b}.
\]
The second is nonpositive by Cauchy--Schwarz.  After subtracting \(2I\), the
matrix therefore has at most one positive eigenvalue.  The multiplication
above preserves the numbers of positive and negative eigenvalues, so the
same conclusion holds for \(\nabla^2Q_{rs}\).  By
\cref{lem:quadratic-stability}, \(Q_{rs}\) is real stable.  Products in
disjoint variables preserve stability, so \(p\) is real stable.

\medskip\noindent\emph{The column selector.}
Define
\begin{equation}\label{eq:column-selector}
  R(z)=\prod_{j=1}^n\left(\sum_{C\in\cG}z_{Cj}\right).
\end{equation}
It is a product of positive linear forms and is therefore multiaffine and real
stable.  Both \(p\) and \(R\) are homogeneous of degree \(n\).  Their
coefficient inner product is exactly the permanent:
\begin{equation}\label{eq:coefficient-is-permanent}
  \langle p,R\rangle=\per(A).
\end{equation}
Indeed, a common squarefree monomial chooses one cluster for every column and
assigns to each cluster as many columns as it has rows.  A singleton has one
possible internal assignment.  The two terms in the coefficient of
\(z_{Cj}z_{Ck}\) for a pair \(C=\set{r,s}\) enumerate the two bijections from
its rows to its selected columns.  Thus the common monomials and their
coefficients enumerate all permutations, with their correct \(A\)-weights.

\medskip\noindent\emph{Applying the coefficient inequality.}
Define a vector indexed by cluster-column pairs by
\begin{equation}\label{eq:cluster-alpha}
  \alpha_{\set{i},j}=X_{ij},
  \qquad
  \alpha_{\set{r,s},j}=X_{rj}+X_{sj}.
\end{equation}
For each cluster \(C\), the sum of \(\alpha_{Cj}\) over \(j\) equals the degree
of \(p_C\).  For each column \(j\), \(\sum_C\alpha_{Cj}=1\).  We may therefore
apply \cref{thm:stable-pairing} to \(p\) and \(R\).

Capacities factor across disjoint variable sets.  For a singleton cluster,
the weighted arithmetic--geometric mean inequality (AM--GM) gives
\begin{equation}\label{eq:singleton-capacity}
  \capop_{X_i}\left(\sum_jA_{ij}z_j\right)
  =\prod_j\left(\frac{A_{ij}}{X_{ij}}\right)^{X_{ij}}.
\end{equation}
For a fixed column \(j\), another application of weighted AM--GM gives
\[
  \capop_{(\alpha_{Cj})_C}\left(\sum_Cz_{Cj}\right)
  =\prod_C\alpha_{Cj}^{-\alpha_{Cj}}.
\]
Consequently
\begin{equation}\label{eq:selector-capacity}
  \capop_\alpha(R)
  =\prod_{C,j}\alpha_{Cj}^{-\alpha_{Cj}}.
\end{equation}

The factor \(\prod_{C,j}\alpha_{Cj}^{\alpha_{Cj}}\) in
\cref{eq:stable-pairing} cancels \cref{eq:selector-capacity}.  What remains
from a singleton cluster is exactly \(S_i(A,X)\), by
\cref{eq:singleton,eq:singleton-capacity}; what remains from a pair is exactly
\(S^{\mathrm{pair}}_{rs}(A,X)\).  Combining
\cref{eq:stable-pairing,eq:coefficient-is-permanent} yields
\[
  \per(A)
  \ge
  \prod_{\set{i}\in\cG}S_i(A,X)
  \prod_{\set{r,s}\in M}S^{\mathrm{pair}}_{rs}(A,X).
\]
Using \cref{eq:singleton-product,eq:pair-certificate} gives
\cref{eq:paired-lower-bound}.
\end{proof}

\section{An exact identity for the Bethe slack}\label{sec:slack}

We now ask what a matrix must look like when the upper bound in
\cref{eq:bethe-sandwich} is nearly tight.  The starting point is a probability
distribution on the permutations counted by the permanent.  Comparing this
distribution with a sequential sampling procedure recovers the
Anari--Rezaei upper bound; keeping the full comparison gives the additional
structural information that we need.

Let
\begin{equation}\label{eq:gibbs-law}
  \mu(\sigma)
  =\frac{\prod_iA_{i,\sigma(i)}}{\per(A)}
\end{equation}
be the Gibbs distribution on permutations, and let
\[
  P_{ij}=\P_{\sigma\sim\mu}{\sigma(i)=j}
\]
be its matrix of assignment marginals: \(P_{ij}\) is the probability that row
\(i\) is assigned to column \(j\).  Since \(A\) is positive, \(\mu\) has full
support and \(P\) is an interior point of \(\cB_n\).

For a probability vector \(p\), define
\begin{align}
  T(p)
  &=\E*_\pi{\sum_jp_j
    \log\left(\sum_{k\ge_\pi j}p_k\right)},\label{eq:row-T}\\
  g(p)
  &=T(p)-\sum_j(1-p_j)\log(1-p_j),\label{eq:row-g}\\
  d(p)
  &=\frac12\log2-g(p),\label{eq:row-deficit}
\end{align}
where \(\pi\) is a uniformly random ordering of the coordinates, and
\(k\ge_\pi j\) means that coordinate \(k\) appears at or after coordinate
\(j\) in this ordering.  Thus \(T(p)\) averages the log suffix mass seen from
a sample drawn from \(p\).  The quantity \(g(p)\) is the resulting row
correction in the sequential comparison, and \(d(p)\) measures its deficit
from the sharp one-row bound
\begin{equation}\label{eq:row-inequality}
  d(p)\ge0
  \qquad\text{for every probability vector }p
\end{equation}
\cite{AnariRezaei2024}.

For an ordering \(\pi=(i_1,\ldots,i_n)\) of the rows, define a sequential
distribution \(\nu_\pi\) on permutations as follows.  Process the rows in the
order \(\pi\).  When row \(i_t\) is processed, choose an unused column \(j\)
with probability proportional to \(P_{i_tj}\).  Let
\begin{equation}\label{eq:averaged-KL}
  D=\E_\pi{\KL(\mu\Vert\nu_\pi)}.
\end{equation}

Using only the inequality \(D\ge0\) is one of the steps that yields the Bethe
upper bound.  Here we retain the value of \(D\) exactly.

\begin{lemma}[Exact sequential identity]\label{lem:exact-sequential}
We have
\begin{equation}\label{eq:exact-sequential}
  \log\per(A)
  =\beta_A(P)+\sum_i g(P_i)-D.
\end{equation}
\end{lemma}

\begin{proof}
Fix a row ordering \(\pi\) and a permutation \(\sigma\).  When row \(i\) is
processed, the unused columns are precisely the columns assigned by \(\sigma\)
to rows that appear at or after \(i\) in \(\pi\).  Thus
\[
  \log\nu_\pi(\sigma)
  =\sum_i\log P_{i,\sigma(i)}
   -\sum_i
    \log\left(
      \sum_{\substack{k:\,
        \sigma^{-1}(k)\text{ is at or after }i\text{ in }\pi}}
      P_{ik}
    \right).
\]
On the other hand,
\[
  \E_\mu{\log\mu(\sigma)}
  =\sum_{i,j}P_{ij}\log A_{ij}-\log\per(A).
\]
For fixed \(\sigma\), a uniformly random row ordering induces a uniformly
random ordering of the columns through \(\sigma\).  Averaging the preceding
display over \(\pi\) and then over \(\sigma\sim\mu\) therefore gives
\begin{align*}
  D
  &=
  \sum_{i,j}P_{ij}\log A_{ij}-\log\per(A)
  -\sum_{i,j}P_{ij}\log P_{ij}
  +\sum_iT(P_i).
\end{align*}
Rearranging, and using
\[
  T(P_i)
  =g(P_i)+\sum_j(1-P_{ij})\log(1-P_{ij}),
\]
gives \cref{eq:exact-sequential}.
\end{proof}

Define the slack in the Bethe upper bound and the Bethe suboptimality of \(P\)
by
\begin{equation}\label{eq:slack-definitions}
  \Delta(A)
  =\frac n2\log2+\log\Bethe(A)-\log\per(A),
  \qquad
  E=\log\Bethe(A)-\beta_A(P).
\end{equation}
Thus \(\Delta(A)\) is the gap in the upper half of the Bethe sandwich, while
\(E\) is the loss from evaluating the Bethe objective at \(P\) rather than at
its maximizer.

\begin{lemma}[Exact slack decomposition]\label{lem:slack-decomposition}
We have
\begin{equation}\label{eq:slack-decomposition}
  \Delta(A)=E+\sum_i d(P_i)+D.
\end{equation}
In particular, every term on the right is nonnegative.
\end{lemma}

\begin{proof}
Substitute \cref{eq:exact-sequential} into \cref{eq:slack-definitions} and use
\(d(P_i)=\frac12\log2-g(P_i)\).  The nonnegativity of \(E\) follows from the
definition of the Bethe optimum, the nonnegativity of the row deficits is
\cref{eq:row-inequality}, and \(D\ge0\) is Gibbs' inequality.
\end{proof}

\section{Dimension-free stability of the row inequality}
\label{sec:row-stability}

To use the slack decomposition, we need a dimension-free description of rows
with small deficit.  The equality examples in the row inequality are the
half--half vectors.  For
\(m\ge2\), let
\begin{equation}\label{eq:half-half-set}
  \cE_m
  =\set*{\frac12e_a+\frac12e_b\given a,b\in[m],\ a\ne b}.
\end{equation}
Here \(e_a\) is the \(a\)-th standard basis vector, and
\[
  \dist_1(p,\cE_m)
  =\min_{q\in\cE_m}\sum_j\abs{p_j-q_j}.
\]
The next lemma gives a modulus that is independent of the ambient dimension.

\begin{lemma}[Row stability]\label{lem:row-stability}
There is a universal constant \(C_{\mathrm{stab}}<\infty\) such that every
probability vector \(p\in\R_{\ge0}^m\), \(m\ge2\), satisfies
\begin{equation}\label{eq:row-stability}
  \dist_1(p,\cE_m)\le C_{\mathrm{stab}}d(p)^{1/4}.
\end{equation}
\end{lemma}

\begin{proof}
The proof has two steps.  First we lower-bound \(d(p)\) by a separable
function of the coordinates and identify its equality cases.  We then turn
this lower bound into a quantitative estimate on the two largest coordinates.

\medskip\noindent\emph{A separable lower bound.}
For a random ordering \(\pi\), put
\[
  L_i=\sum_{j\le_\pi i}p_j,
  \qquad
  R_i=\sum_{j\ge_\pi i}p_j.
\]
Pairing every ordering with its reverse gives
\begin{equation}\label{eq:g-symmetrized}
  g(p)
  =\frac12\sum_i p_i\E_\pi{\log(L_iR_i)}
   -\sum_i(1-p_i)\log(1-p_i).
\end{equation}
Let \(s_k=\sum_ip_i^k\).  Averaging over the relative positions of three
indices gives
\begin{equation}\label{eq:LiRi-moment}
  M_i:=\E_\pi{L_iR_i}
  =\frac{1-s_2}{6}+\frac23p_i+\frac13p_i^2.
\end{equation}
To see this, write
\[
  L_iR_i
  =p_i+\left(\sum_{j<_\pi i}p_j\right)
          \left(\sum_{k>_\pi i}p_k\right).
\]
For distinct \(j,k\ne i\), the probability that \(j<_\pi i<_\pi k\) is
\(1/6\).  Summing the corresponding cross terms proves
\cref{eq:LiRi-moment}.

Jensen's inequality and the tangent inequality
\[
  \log x\le-\log2+2x-1
\]
give
\begin{align*}
  g(p)
  &\le
  -\frac12\log2
  +\sum_i p_iM_i-\frac12
  -\sum_i(1-p_i)\log(1-p_i)\\
  &=
  -\frac12\log2-\frac13
  +\frac12s_2+\frac13s_3
  -\sum_i(1-p_i)\log(1-p_i).
\end{align*}
Consequently
\begin{equation}\label{eq:separable-defect}
  d(p)\ge C-\sum_iF(p_i),
\end{equation}
where
\begin{equation}\label{eq:F-definition}
  F(x)=-(1-x)\log(1-x)+\frac{x^2}{2}+\frac{x^3}{3},
  \qquad
  C=\log2+\frac13.
\end{equation}

\medskip\noindent\emph{The equality case.}
Put \(h(x)=F(x)/x\) for \(x>0\), with \(h(0)=1\).  Direct differentiation
gives
\[
  F''(x)=\frac{x(1-2x)}{1-x},
  \qquad
  h'(x)=\frac1{x^2}\int_0^x tF''(t)\,dt.
\]
Thus \(h\) is strictly increasing on \([0,1/2]\), and a direct substitution
gives \(h(1/2)=C\).

Suppose first that every coordinate of \(p\) is at most \(1/2\).  Then
\[
  \sum_iF(p_i)=\sum_ip_ih(p_i)\le C.
\]
Equality is possible only if all positive coordinates equal \(1/2\), hence
only when \(p\in\cE_m\).

Now suppose that the largest coordinate is \(a>1/2\), and put \(q=1-a\).
Every remaining coordinate is at most \(q\le1/2\), so monotonicity of \(h\)
gives
\[
  \sum_{i:p_i\ne a}F(p_i)\le F(q).
\]
For
\[
  J_{\mathrm{row}}(a)=C-F(a)-F(1-a)
\]
and \(u=2a-1\), we have
\begin{equation}\label{eq:J-derivative}
  J_{\mathrm{row}}'(a)
  =\log\frac{a}{1-a}-4a+2
  =2\left(\operatorname{artanh}u-u\right)\ge0.
\end{equation}
Here \(\operatorname{artanh}\) denotes the inverse hyperbolic tangent.
Moreover,
\begin{equation}\label{eq:quartic-defect}
  J_{\mathrm{row}}(a)
  =\int_0^{2a-1}\left(\operatorname{artanh}t-t\right)\,dt
  \ge\frac{(2a-1)^4}{12}
  =\frac43\left(a-\frac12\right)^4.
\end{equation}
Thus the separable defect is nonnegative, vanishes only at a half--half
vector, and grows quartically when one coordinate exceeds \(1/2\).

\medskip\noindent\emph{Quantitative stability.}
Let \(a\ge b\) be the two largest
coordinates of \(p\), and put \(t=1-a-b\).  Define
\[
  c_1=\min_{x\in[1/4,1/2]}h'(x)>0,
  \qquad
  c_2=\min_{x\in[1/4,1/2]}\left(h(x)-h(x/2)\right)>0,
\]
and
\begin{equation}\label{eq:delta-star}
  \delta_\star
  =\min\set*{
    \frac1{192},
    \frac12\left(C-h(1/4)\right)
  }>0.
\end{equation}
All three constants are absolute.

Assume first that \(a\le1/2\) and \(d(p)\le\delta_\star\).  Since the
separable defect in \cref{eq:separable-defect} is at most \(d(p)\), we must
have \(a,b\ge1/4\): if \(a<1/4\), then the defect is at least
\(C-h(1/4)\); if \(b<1/4\), the mass outside the largest coordinate is at
least \(1/2\), and its contribution to the defect is at least
\(\frac12(C-h(1/4))\).  Monotonicity on \([1/4,1/2]\) now gives
\[
  \frac12-a\le\frac{4d(p)}{c_1},
  \qquad
  \frac12-b\le\frac{4d(p)}{c_1}.
\]
Indeed, the coordinates \(a\) and \(b\) contribute at least
\(a c_1(1/2-a)\) and \(b c_1(1/2-b)\), respectively, to the separable
defect, and both weights are at least \(1/4\).  Since
\(t=(1/2-a)+(1/2-b)\),
\begin{equation}\label{eq:stability-below-half}
  \dist_1(p,\cE_m)=2t\le\frac{16d(p)}{c_1}.
\end{equation}

Suppose next that \(a>1/2\) and \(d(p)\le\delta_\star\).  By
\cref{eq:quartic-defect},
\[
  a-\frac12\le\left(\frac{3d(p)}4\right)^{1/4}\le\frac14,
\]
so \(q=1-a\ge1/4\).  Write the coordinates other than \(a\) as
\(x_1,\ldots,x_r\), with \(b=\max_\ell x_\ell\).  The remaining separable
defect is
\begin{equation}\label{eq:tail-separable-defect}
  F(q)-\sum_\ell F(x_\ell)
  =\sum_\ell x_\ell\left(h(q)-h(x_\ell)\right).
\end{equation}
If \(t=q-b\le q/2\), all coordinates other than \(b\) are at most \(q/2\),
and their total mass is \(t\).  If \(t>q/2\), then every tail coordinate is at
most \(q/2\), and their total mass \(q\) is at least \(t\).  In both cases
\cref{eq:tail-separable-defect} is at least \(c_2t\).  Therefore
\begin{equation}\label{eq:stability-above-half}
  \dist_1(p,\cE_m)
  =2\left(a-\frac12\right)+2t
  \le
  2\left(\frac{3d(p)}4\right)^{1/4}
  +\frac{2d(p)}{c_2}.
\end{equation}

For \(d(p)>\delta_\star\), use the bound
\(\dist_1(p,\cE_m)\le2\).  Thus \cref{eq:row-stability} holds with, for
example,
\[
  C_{\mathrm{stab}}=\max\set*{
    \frac{16}{c_1},
    2\left(\frac34\right)^{1/4}+\frac2{c_2},
    2\delta_\star^{-1/4}
  }.
\]
\end{proof}

\section{Robust four-cycle structure}\label{sec:cycles}

Row stability makes most rows close to half--half.  We encode their two large
coordinates as heavy edges.  The column constraints give a graph of maximum
degree two, and the sequential divergence controls the rows lying on long
cycles.

The key observation is entropic.  Once we know which rows leave the heavy-edge
graph, each cycle carries at most one bit of matching ambiguity.  A good row,
on the other hand, contributes almost half a bit to the entropy calculation
below.
Cycles with at least three rows therefore create a definite surplus, which
must be paid for by the divergence \(D\).  Since our logarithms are natural,
one bit here means \(\log2\) units of entropy.

Fix \(0<\eta\le1/10\).  Call row \(i\) \emph{good} if
\[
  \dist_1(P_i,\cE_n)\le\eta,
\]
and let \(b\) be the number of bad rows.  By
\cref{lem:slack-decomposition,lem:row-stability}, every bad row has deficit at
least
\[
  d_0=\left(\frac{\eta}{C_{\mathrm{stab}}}\right)^4,
\]
and therefore
\begin{equation}\label{eq:bad-row-count}
  b\le\frac{\Delta(A)}{d_0}.
\end{equation}
For every good row, choose two coordinates witnessing its distance from
\(\cE_n\).  If their masses are \(u,v\) and the remaining mass is \(q\), then
\begin{equation}\label{eq:good-row-coordinates}
  \abs{u-\frac12}+\abs{v-\frac12}+q\le\eta.
\end{equation}
In particular, both distinguished entries are at least \(1/2-\eta\), and the
mass outside them is at most \(\eta\).

\subsection{The heavy-edge graph}

Form the bipartite graph
\begin{equation}\label{eq:heavy-graph}
  F_\eta
  =\set*{(i,j)\in[n]\times[n]\given P_{ij}\ge\frac12-\eta}.
\end{equation}
Its two vertex sets are the rows and columns of \(A\).
Every row and column has degree at most two, since
\(3(1/2-\eta)>1\), and every good row has degree exactly two.

We complete \(F_\eta\) to a spanning bipartite two-regular multigraph \(K\).
This can be done without deleting an edge: attach \(2-\deg(v)\) stubs to every
row and column vertex, and pair the row stubs arbitrarily with the column
stubs.  The two sides have the same number of stubs because they have the same
number of vertices and the same total degree in \(F_\eta\).  Parallel edges
are allowed.  Every component of \(K\) is therefore an even cycle, where a
pair of parallel edges is allowed as a cycle containing one row.  Such a
doubled edge contains no good row, since a good row already has two distinct
neighbors in \(F_\eta\).

Let \(N_K(i)\) denote the set of column neighbors of row \(i\) in \(K\), with
parallel edges ignored.  Let \(N\) be the number of good rows in components of
\(K\) containing at least three rows.

\subsection{An entropy encoding}

View a sampled permutation \(\sigma\sim\mu\) as a perfect matching between
the row and column vertices.  Define
\[
  Y_i=
  \begin{cases}
    \star,&\sigma(i)\in N_K(i),\\
    \sigma(i),&\sigma(i)\notin N_K(i).
  \end{cases}
\]
Thus \(Y_i\) records the exact column used by row \(i\) when its matching edge
is not in \(K\), and otherwise only records that it uses one of its neighbors
in \(K\).
Let \(m(K)\) be the number of components of \(K\) containing at least two
rows.  The next lemma formalizes the fact that each surviving cycle costs at
most one bit.

\begin{lemma}[Core encoding]\label{lem:core-encoding}
We have
\begin{equation}\label{eq:core-encoding}
  H(\mu)\le\sum_iH(Y_i)+m(K)\log2.
\end{equation}
\end{lemma}

\begin{proof}
Fix a value \(y\) of the vector \(Y=(Y_1,\ldots,Y_n)\).  It determines the
set of escaped rows and the columns assigned to them.  Delete these rows and
columns from \(K\).  Every remaining assignment consistent with \(y\) is a
perfect matching of the resulting subgraph.

Within an original component of \(K\), deleting vertices either leaves the
whole cycle intact or breaks it into paths.  An intact cycle containing at
least two rows has two perfect matchings.  Every path that has a perfect
matching has a unique one, obtained greedily from an endpoint.  A doubled
edge represents only one row-column assignment.  Consequently there are at
most \(2^{m(K)}\) permutations consistent with \(y\), and
\[
  H(\sigma\mid Y)\le m(K)\log2.
\]
Now use
\[
  H(\mu)=H(\sigma)\le H(Y)+H(\sigma\mid Y)
  \le\sum_iH(Y_i)+m(K)\log2.
\]
\end{proof}

The KL term from \cref{eq:averaged-KL} has a useful entropy form.  For a
probability vector \(p\), put
\begin{equation}\label{eq:row-score}
  s(p)=H(p)+T(p).
\end{equation}

\begin{lemma}[Entropy form of the sequential divergence]
\label{lem:entropy-KL}
We have
\begin{equation}\label{eq:entropy-KL}
  D=-H(\mu)+\sum_i s(P_i).
\end{equation}
\end{lemma}

\begin{proof}
The Gibbs law gives
\[
  H(\mu)=\log\per(A)-\sum_{i,j}P_{ij}\log A_{ij}.
\]
Substitute \cref{eq:exact-sequential} and expand the Bethe objective.  For a
row \(p\), the terms that remain are
\[
  -\sum_jp_j\log p_j
  +\sum_j(1-p_j)\log(1-p_j)+g(p)
  =H(p)+T(p).
\]
Summing over the rows proves the identity.
\end{proof}

\subsection{The score of a good row}

A good row contributes almost \((\log2)/2\) beyond the information already
contained in \(Y_i\).  The following Riemann-sum identity makes the dependence
on the leakage explicit.

\begin{lemma}[Suffix Riemann-sum identity]\label{lem:suffix-identity}
Fix an ordering of a probability vector \(p\), write its coordinates in that
order as \(p_1,\ldots,p_m\), and put
\[
  z_r=\sum_{k=r}^m p_k,
  \qquad z_{m+1}=0.
\]
For
\begin{equation}\label{eq:e-definition}
  e(a,x)=a-x\log(1+a/x),
  \qquad e(a,0)=a,
\end{equation}
we have
\begin{equation}\label{eq:suffix-identity}
  \sum_{r=1}^mp_r\log z_r
  =-1+\sum_{r=1}^m e(p_r,z_{r+1}).
\end{equation}
Moreover, \(e(a,x)\ge0\), and for fixed \(a\ge0\), the function is
nonincreasing in \(x\ge0\).
\end{lemma}

\begin{proof}
On the interval \([z_{r+1},z_r]\),
\begin{align*}
  p_r\log z_r-\int_{z_{r+1}}^{z_r}\log x\,dx
  &=p_r-z_{r+1}\log\left(1+\frac{p_r}{z_{r+1}}\right)\\
  &=e(p_r,z_{r+1}).
\end{align*}
Summing over \(r\), and using \(\int_0^1\log x\,dx=-1\), proves
\cref{eq:suffix-identity}.  The inequality
\(\log(1+t)\le t\) gives \(e(a,x)\ge0\).  Finally,
\[
  \frac{\partial}{\partial x}e(a,x)
  =-\log(1+a/x)+\frac{a}{a+x}\le0,
\]
where the last inequality follows from
\(\log(1+t)\ge t/(1+t)\).
\end{proof}

Consider one good row.  Call its two distinguished masses \(u,v\), call the
remaining mass \(q=1-u-v\), and refer to the two distinguished coordinates
as its core coordinates.
The random variable \(Y_i\) merges the two core outcomes and leaves every
outside outcome distinct.  Hence
\begin{equation}\label{eq:entropy-coarsening-exact}
  H(P_i)-H(Y_i)
  =(1-q)\Htwo\left(\frac{u}{1-q}\right).
\end{equation}
In \cref{eq:suffix-identity}, keep only the two nonnegative summands
corresponding to the core coordinates.  With probability \(1/2\), the other
core coordinate precedes \(u\); then the suffix mass after \(u\) is at most
\(q\).  With probability \(1/2\), it follows \(u\); then that suffix mass is
at most \(v+q\).  Monotonicity of \(e\) gives the analogous two estimates for
\(v\).  Averaging over the ordering yields
\begin{align}
  s(P_i)-H(Y_i)
  &\ge\Psi(u,v,q),\label{eq:Psi-bound}\\
  \Psi(u,v,q)
  &=(1-q)\Htwo\left(\frac{u}{1-q}\right)-1\notag\\
  &\quad+\frac12\left[
    e(u,q)+e(u,v+q)+e(v,q)+e(v,u+q)
  \right].\notag
\end{align}

At \((u,v,q)=(1/2,1/2,0)\), the right side equals
\((\log2)/2\).  Continuity makes this estimate uniform over all good rows.

\begin{lemma}[Good-row score]\label{lem:good-row-score}
There is a nondecreasing function \(\omega:[0,1/10]\to\R_{\ge0}\), with
\(\omega(\eta)\to0\) as \(\eta\downarrow0\), such that every good row
satisfies
\begin{equation}\label{eq:good-row-score}
  s(P_i)-H(Y_i)
  \ge\frac12\log2-\omega(\eta).
\end{equation}
\end{lemma}

\begin{proof}
The function \(e\) is continuous on \(\R_{\ge0}^2\) under the convention
\(e(a,0)=a\).  Hence \(\Psi\) is continuous in a neighborhood of
\((1/2,1/2,0)\), where
\[
  \Psi(1/2,1/2,0)=\frac12\log2.
\]
For \(0\le\eta\le1/10\), let \(\cD_\eta\) be the compact set of triples
\((u,v,q)\in\R_{\ge0}^3\) satisfying
\[
  u+v+q=1,
  \qquad
  \abs{u-\frac12}+\abs{v-\frac12}+q\le\eta,
\]
and define
\[
  \omega(\eta)
  =\max_{(u,v,q)\in\cD_\eta}
   \left(\frac12\log2-\Psi(u,v,q)\right)_+.
\]
Here \((x)_+=\max\set{0,x}\).
These sets are nested and shrink to \(\set{(1/2,1/2,0)}\), so continuity
implies that \(\omega\) is nondecreasing and tends to zero at the origin.
The claim follows from \cref{eq:good-row-coordinates,eq:Psi-bound}.
\end{proof}

For a bad row, the coarser estimate suffices.  By
\cref{lem:suffix-identity}, \(T(p)\ge-1\) for every probability vector \(p\).
Since \(Y_i\) is a coarsening of an outcome sampled from \(P_i\),
\(H(P_i)\ge H(Y_i)\).  Therefore
\begin{equation}\label{eq:bad-row-score}
  s(P_i)-H(Y_i)\ge-1
  \qquad\text{for every bad row }i.
\end{equation}

\subsection{Excluding long components}

We can now compare the half-bit contribution of each good row with the
one-bit ambiguity of each cycle.

\begin{lemma}[Robust cycle information]\label{lem:robust-cycle}
For the fixed threshold \(\eta\), let \(b\) be the number of bad rows and let
\(N\) be the number of good rows in components of \(K\) containing at least
three rows.  Then
\begin{equation}\label{eq:robust-cycle}
  D
  \ge
  \frac{\log2}{6}N
  -\left(1+\frac{\log2}{2}\right)b
  -n\omega(\eta).
\end{equation}
\end{lemma}

\begin{proof}
Let \(G\) be the set of good rows.  Combining
\cref{eq:entropy-KL,eq:core-encoding,eq:good-row-score,eq:bad-row-score}
gives
\begin{equation}\label{eq:D-component-start}
  D
  \ge
  \abs{G}\left(\frac12\log2-\omega(\eta)\right)
  -b-m(K)\log2.
\end{equation}
For a component \(C\) of \(K\), let \(k_C,g_C,b_C\) be its number of rows,
good rows, and bad rows.  A one-row doubled component has \(g_C=0\).  If
\(k_C=2\), then
\[
  \frac{g_C}{2}-1=-\frac{b_C}{2}.
\]
If \(k_C\ge3\), then
\[
  \frac{g_C}{2}-1
  \ge\frac{g_C}{6}-\frac{b_C}{3},
\]
  because the difference between the two sides is \((k_C-3)/3\).  Summing over
the components, including the one-row components, gives
\begin{equation}\label{eq:component-accounting}
  \frac{\abs{G}}2-m(K)
  \ge\frac N6-\frac b2.
\end{equation}
Substitute \cref{eq:component-accounting} into
\cref{eq:D-component-start}, and use \(\abs{G}\le n\).
\end{proof}

A component containing exactly two rows is a \(K_{2,2}\).  If such a component
contains a bad row, it contains at most one good row, so contaminated
two-row components account for at most \(b\) good rows.  The number of rows in
two-row components whose two rows are good is therefore at least
\begin{equation}\label{eq:clean-cycle-count}
  n-2b-N.
\end{equation}
Equivalently, \(K\) contains at least \((n-2b-N)/2\) vertex-disjoint
\emph{clean pairs}: pairs of good rows forming a \(K_{2,2}\) component.

\section{Transfer to a computable Bethe point}\label{sec:transfer}

The structural conclusions of \cref{sec:cycles} concern the true marginal
matrix \(P\), which is not available to the algorithm.  We transfer these
conclusions to a point obtained by concave optimization.

The entropy regularizer serves two purposes: it places the optimizer in the
interior of the Birkhoff polytope, and its first-order optimality equations
expose coordinates in which the paired gain factorizes.  The main identity in
this section then
measures how far those coordinates can move on the heavy edges of \(P\).

Fix a rational constant \(\xi>0\).  For \(n\ge2\), let
\begin{equation}\label{eq:tau-definition}
  \ell_n=\lceil\log_2n\rceil,
  \qquad
  \tau=\frac{\xi}{4\ell_n}.
\end{equation}
Define the regularized Bethe objective
\begin{equation}\label{eq:regularized-objective}
  \Phi_\tau(Z)
  =\beta_A(Z)+\tau\sum_iH(Z_i),
  \qquad Z\in\cB_n,
\end{equation}
and let \(X\) be its maximizer.

\begin{lemma}[Regularized optimizer]\label{lem:regularized-optimizer}
The maximizer \(X\) is unique and lies in the interior of \(\cB_n\).  Moreover,
\begin{equation}\label{eq:near-bethe}
  \beta_A(X)\ge\log\Bethe(A)-\xi n.
\end{equation}
There are positive row and column scalings \(r_i,c_j\) such that
\begin{equation}\label{eq:kkt}
  A_{ij}=r_ic_jX_{ij}^{1+\tau}(1-X_{ij})
  \qquad\forall i,j.
\end{equation}
\end{lemma}

\begin{proof}
The Bethe objective is concave on \(\cB_n\), and Shannon entropy is strictly
concave.  Thus \(\Phi_\tau\) is strictly concave.

To prove interiority, suppose that \(Z\in\cB_n\) is on the boundary.  Let
\(W=J/n\) and put \(Z(t)=(1-t)Z+tW\).  A boundary point has a zero entry,
say \(Z_{ij}=0\).  Concavity of the Bethe objective gives
\[
  \beta_A(Z(t))
  \ge (1-t)\beta_A(Z)+t\beta_A(W).
\]
Coordinatewise concavity of \(x\mapsto-x\log x\) gives the analogous bound
for the row entropies.  At the chosen zero entry, the difference from this
chord is exactly
\[
  -\frac tn\log\frac tn
  -t\left(-\frac1n\log\frac1n\right)
  =\frac tn\log\frac1t.
\]
Consequently
\[
  \Phi_\tau(Z(t))-\Phi_\tau(Z)
  \ge
  t\bigl(\Phi_\tau(W)-\Phi_\tau(Z)\bigr)
  +\frac{\tau t}{n}\log\frac1t.
\]
The last term dominates the first as \(t\downarrow0\).  Thus every boundary
point can be improved, and the maximizer is interior.

Since \(\sum_iH(Z_i)\le n\log n\) and
\(\log n\le\ell_n\log2\), the regularizer is at most
\(\xi n\).  Comparing \(X\) with a Bethe maximizer gives
\cref{eq:near-bethe}.

At the interior maximizer, the Karush--Kuhn--Tucker (KKT) conditions apply.
The derivative of the objective in coordinate
\((i,j)\) is
\[
  \log A_{ij}-(1+\tau)\log X_{ij}
  -\log(1-X_{ij})-(2+\tau).
\]
The KKT equations say that this matrix is a sum of a row potential and a
column potential.  Exponentiating, and absorbing the constant \(2+\tau\)
into the potentials, gives \cref{eq:kkt}.
\end{proof}

For a row \(i\), put
\begin{equation}\label{eq:U-definition}
  q_i=\prod_k(1-X_{ik}),
  \qquad
  U_{ij}
  =\frac{X_{ij}^{1+\tau}(1-X_{ij})}{q_i}
  =\frac{X_{ij}^{1+\tau}}
         {\prod_{k\ne j}(1-X_{ik})}.
\end{equation}
The normalization by \(q_i\) makes the variables \(U_{ij}\) both bounded and
compatible with the KKT factorization.  These are the two properties used
below.

\begin{lemma}[Elementary bounds on \(U\)]\label{lem:U-bounds}
For every row \(i\),
\begin{equation}\label{eq:U-bounds}
  0<U_{ij}\le1,
  \qquad
  \sum_jU_{ij}\le e.
\end{equation}
\end{lemma}

\begin{proof}
Suppress the row index and write \(x_j=X_{ij}\) and
\(q=\prod_k(1-x_k)\).  The elementary product
inequality gives
\[
  \prod_{k\ne j}(1-x_k)
  \ge1-\sum_{k\ne j}x_k=x_j.
\]
Therefore
\[
  U_j\le x_j^\tau\le1.
\]
For the row sum, first discard \(x_j^\tau\le1\):
\begin{equation}\label{eq:U-sum-first}
  \sum_jU_j
  \le\frac{\sum_jx_j(1-x_j)}{\prod_k(1-x_k)}
  =\frac{1-s_2}{q},
  \qquad s_2=\sum_jx_j^2.
\end{equation}
Put \(a=\sqrt{s_2}<1\).  Monotonicity of \(\ell_p\)-norms gives
\(\sum_jx_j^k\le a^k\) for \(k\ge2\).  Hence
\begin{align*}
  -\log q
  &=1+\sum_{k\ge2}\frac{\sum_jx_j^k}{k}\\
  &\le1+\sum_{k\ge2}\frac{a^k}{k}\\
  &=1-\log(1-a)-a\\
  &\le1-\log(1-a^2),
\end{align*}
where the last inequality is equivalent to \(\log(1+a)\le a\).
Thus \(q\ge e^{-1}(1-s_2)\), and \cref{eq:U-sum-first} proves the result.
\end{proof}

\subsection{The transfer identity}

The transfer identity is the bridge between the structural point \(P\) and
the computable point \(X\).  It assigns a nonnegative cost to every entry of
\(U\).  We first bound the total cost and then subtract the part caused by
assignments outside \(K\), leaving a bound on the heavy edges.

For a probability vector \(p\), define its row contribution to the Bethe
entropy by
\begin{equation}\label{eq:hB-definition}
  h_B(p)=H(p)+\sum_j(1-p_j)\log(1-p_j),
\end{equation}
and let
\[
  E_\tau=\Phi_\tau(X)-\Phi_\tau(P)\ge0.
\]
Thus \(E_\tau\) is the suboptimality of the marginal matrix \(P\) in the
regularized program.

\begin{lemma}[Global transfer identity]\label{lem:global-transfer}
We have the exact identity
\begin{equation}\label{eq:global-transfer}
  E_\tau+\sum_i\left[h_B(P_i)+\tau H(P_i)\right]
  =\sum_{i,j}P_{ij}\log\frac1{U_{ij}}.
\end{equation}
Moreover,
\begin{equation}\label{eq:global-transfer-upper}
  \sum_{i,j}P_{ij}\log\frac1{U_{ij}}
  \le
  \Delta(A)+2\xi n+H(\mu)-\frac n2\log2.
\end{equation}
\end{lemma}

\begin{proof}
Write \(\log A_{ij}\) using \cref{eq:kkt}.  Row and column scaling terms
cancel between \(\Phi_\tau(X)\) and \(\Phi_\tau(P)\).  At \(X\), all remaining
terms sum to \(\sum_i\log q_i\).  Expanding the difference gives
\begin{align*}
  E_\tau
  &=
  \sum_i\log q_i
  -(1+\tau)\sum_{i,j}P_{ij}\log X_{ij}
  -\sum_{i,j}P_{ij}\log(1-X_{ij})\\
  &\quad
  +(1+\tau)\sum_{i,j}P_{ij}\log P_{ij}
  -\sum_{i,j}(1-P_{ij})\log(1-P_{ij}).
\end{align*}
Adding \(\sum_i[h_B(P_i)+\tau H(P_i)]\) cancels the second line.  The first
line is exactly \(\sum_{i,j}P_{ij}\log(1/U_{ij})\), proving
\cref{eq:global-transfer}.

To obtain the upper bound, use that the regularizer is nonnegative and at
most \(\xi n\):
\[
  0\le E_\tau\le E+\xi n,
  \qquad
  \tau\sum_iH(P_i)\le\xi n.
\]
By \cref{eq:entropy-KL} and \(s=g+h_B\),
\[
  \sum_i h_B(P_i)
  =H(\mu)+D-\sum_i g(P_i).
\]
Finally,
\[
  \sum_i g(P_i)
  =\frac n2\log2-\sum_i d(P_i).
\]
Use \cref{eq:slack-decomposition} to replace
\(E+D+\sum_i d(P_i)\) by \(\Delta(A)\).  This gives
\cref{eq:global-transfer-upper}.
\end{proof}

For each row \(i\), let
\begin{equation}\label{eq:rho-i}
  \rho_i=\sum_{j\notin N_K(i)}P_{ij}.
\end{equation}
Thus \(\rho_i\) is the probability that the matching edge of row \(i\) lies
outside \(K\).  For a good row, \(\rho_i\le\eta\).  We call the edges from
\(i\) to \(N_K(i)\) its core edges and define their transfer cost by
\begin{equation}\label{eq:R-core}
  R_{\mathrm{core}}
  =\sum_i\sum_{j\in N_K(i)}
    P_{ij}\log\frac1{U_{ij}}.
\end{equation}

\begin{lemma}[Excursion-entropy cancellation]\label{lem:core-transfer}
We have
\begin{equation}\label{eq:core-transfer}
  R_{\mathrm{core}}
  \le
  \Delta(A)+2\xi n
  +\sum_i\left[\Htwo(\rho_i)+\rho_i\right].
\end{equation}
Consequently,
\begin{equation}\label{eq:core-transfer-normalized}
  \frac{R_{\mathrm{core}}}{n}
  \le
  \frac{\Delta(A)}n+2\xi+\Htwo(\eta)+\eta
  +(1+\log2)\frac bn.
\end{equation}
\end{lemma}

\begin{proof}
Combine \cref{eq:global-transfer-upper,eq:core-encoding} and
\(m(K)\le n/2\) to obtain
\begin{equation}\label{eq:transfer-before-tail}
  \sum_{i,j}P_{ij}\log\frac1{U_{ij}}
  \le\Delta(A)+2\xi n+\sum_iH(Y_i).
\end{equation}
Fix a row and let \(O=[n]\setminus N_K(i)\) be its set of outside columns.
Write \(\rho=\sum_{j\in O}P_{ij}\) and \(W=\sum_{j\in O}U_{ij}\).  If \(\rho=0\),
the desired estimate holds.  Otherwise the log-sum inequality gives
\begin{align*}
  \sum_{j\in O}P_{ij}\log\frac1{U_{ij}}
  &=
  \rho H\left((P_{ij}/\rho)_{j\in O}\right)
  +\sum_{j\in O}P_{ij}
    \log\frac{P_{ij}/\rho}{U_{ij}/W}
  -\rho\log W\\
  &\ge
  \rho H\left((P_{ij}/\rho)_{j\in O}\right)-\rho,
\end{align*}
where we used \(W\le e\) from \cref{lem:U-bounds}.  Since
\[
  H(Y_i)
  =\Htwo(\rho)
   +\rho H\left((P_{ij}/\rho)_{j\in O}\right),
\]
we obtain
\begin{equation}\label{eq:tail-transfer}
  \sum_{j\notin N_K(i)}P_{ij}\log\frac1{U_{ij}}
  \ge H(Y_i)-\Htwo(\rho_i)-\rho_i.
\end{equation}
Subtract \cref{eq:tail-transfer} from \cref{eq:transfer-before-tail} and sum
over the rows to obtain \cref{eq:core-transfer}.  For good rows use
\(\rho_i\le\eta\); for bad rows use
  \(\Htwo(\rho_i)+\rho_i\le\log2+1\).  This gives the normalized bound.
\end{proof}

\section{A constant gain from a clean four-cycle}\label{sec:gain}

The KKT equations also factor the paired certificate.  For a row \(i\), put
\begin{equation}\label{eq:zeta-i}
  \zeta_i=\prod_jX_{ij}^{\tau X_{ij}}\le1.
\end{equation}

\begin{lemma}[Pair factorization]\label{lem:pair-factorization}
For two rows \(r,s\), let \(\alpha_j=X_{rj}+X_{sj}\).  Then
\begin{equation}\label{eq:pair-factorization}
  \Gamma_{rs}(A,X)
  =\frac1{\zeta_r\zeta_s}
   \prod_j(1-\alpha_j)^{1-\alpha_j}
   \capop_\alpha\left(
     \sum_{j<k}
     (U_{rj}U_{sk}+U_{rk}U_{sj})z_jz_k
   \right).
\end{equation}
\end{lemma}

\begin{proof}
Substituting \cref{eq:kkt} into the singleton factor gives
\begin{equation}\label{eq:S-factorized}
  S_i(A,X)
  =r_iq_i\zeta_i\prod_jc_j^{X_{ij}}.
\end{equation}
For the pair polynomial, \cref{eq:kkt,eq:U-definition} give
\begin{align*}
  Q_{rs}(z)
  =r_rr_sq_rq_s
  \sum_{j<k}c_jc_k
  (U_{rj}U_{sk}+U_{rk}U_{sj})z_jz_k.
\end{align*}
Rescaling the capacity variables by \(c_j\) therefore gives
\begin{align*}
  \capop_\alpha(Q_{rs})
  &=
  r_rr_sq_rq_s
  \left(\prod_jc_j^{\alpha_j}\right)\\
  &\quad\cdot
  \capop_\alpha\left(
    \sum_{j<k}
    (U_{rj}U_{sk}+U_{rk}U_{sj})z_jz_k
  \right).
\end{align*}
Divide by \cref{eq:S-factorized} for rows \(r,s\).  Since
\(\alpha_j=X_{rj}+X_{sj}\), the row scalings, column scalings, and the factors
\(q_rq_s\) cancel, leaving \cref{eq:pair-factorization}.
\end{proof}

A clean pair gives a constant gain whenever its four core entries have small
transfer cost.

\begin{lemma}[Clean-pair gain]\label{lem:clean-gain}
There are rational constants \(\kappa_0,\gamma_0,\xi_0>0\) with the following
property.  Suppose rows \(r,s\) and columns \(a,b\) form a clean \(K_{2,2}\)
component of \(K\).  If \(0<\xi\le\xi_0\) and
\begin{equation}\label{eq:small-local-cost}
  \sum_{i\in\set{r,s}}\sum_{j\in\set{a,b}}
  \log\frac1{U_{ij}}\le\kappa_0,
\end{equation}
then
\begin{equation}\label{eq:clean-gain}
  \log\Gamma_{rs}(A,X)\ge\gamma_0.
\end{equation}
\end{lemma}

\begin{proof}
The proof has two steps.  Large core entries of \(U\) first force both rows of
\(X\) to place almost all their mass on the two core columns.  We then build a
feasible distribution for the entropy dual of capacity which places most of
its mass on the core monomial \(z_az_b\).

Let \(\kappa>0\) be a parameter to be fixed below and put
\(u_\kappa=e^{-\kappa}\).
Assume the left side of \cref{eq:small-local-cost} is at most \(\kappa\).  By
\cref{lem:U-bounds}, every summand in \cref{eq:small-local-cost} is
nonnegative.  Thus all four core entries of \(U\) are at least \(u_\kappa\).

\medskip\noindent\emph{Small transfer cost implies small leakage.}
Consider one of the two rows and write
\[
  x_a+x_b+t=1,
  \qquad
  t=\sum_{j\notin\set{a,b}}X_{ij}.
\]
Since \(\prod_{j\notin\set{a,b}}(1-X_{ij})\ge1-t\), we have
\begin{align}
  U_{ia}
  &\le
  \frac{x_a}{(1-x_b)(1-t)}
  =\frac{x_a}{x_a+tx_b},\label{eq:Ua-leakage}\\
  U_{ib}
  &\le
  \frac{x_b}{x_b+tx_a}.\label{eq:Ub-leakage}
\end{align}
The inequalities \(U_{ia},U_{ib}\ge u_\kappa\) imply
\[
  u_\kappa tx_b\le(1-u_\kappa)x_a,
  \qquad
  u_\kappa tx_a\le(1-u_\kappa)x_b.
\]
Multiplying and using \(x_a,x_b>0\) gives
\begin{equation}\label{eq:one-row-leakage}
  t\le u_\kappa^{-1}-1=e^\kappa-1.
\end{equation}

Let
\[
  \alpha_j=X_{rj}+X_{sj},
  \qquad
  \rho=\sum_{j\notin\set{a,b}}\alpha_j.
\]
Applying \cref{eq:one-row-leakage} to both rows gives
\begin{equation}\label{eq:alpha-leakage}
  \rho\le \overline\rho_\kappa:=2(e^\kappa-1).
\end{equation}
Put \(\delta_a=1-\alpha_a\) and \(\delta_b=1-\alpha_b\).  Since
\(\sum_j\alpha_j=2\), we have \(\delta_a+\delta_b=\rho\).

\medskip\noindent\emph{A capacity witness concentrated on the core.}
Consider the following probability distribution on two-element column sets:
\begin{equation}\label{eq:theta-distribution}
  \theta_{\set{a,b}}=1-\rho,
  \qquad
  \theta_{\set{a,\ell}}=\frac{\delta_b}{\rho}\alpha_\ell,
  \qquad
  \theta_{\set{b,\ell}}=\frac{\delta_a}{\rho}\alpha_\ell
  \quad
  (\ell\notin\set{a,b}).
\end{equation}
At \(\rho=0\), use the limiting distribution concentrated on
\(\set{a,b}\).  For \(\rho>0\), the total mass is one.  Its inclusion
marginal at an outside column \(\ell\) is \(\alpha_\ell\), while
\[
  1-\rho+\delta_b=1-\delta_a=\alpha_a,
  \qquad
  1-\rho+\delta_a=1-\delta_b=\alpha_b.
\]
Thus its mean exponent vector is \(\alpha\), so \(\theta\) is feasible in
\cref{lem:capacity-certificate}.

Let
\[
  \widetilde Q(z)
  =\sum_{j<k}
    (U_{rj}U_{sk}+U_{rk}U_{sj})z_jz_k
\]
be the polynomial in \cref{eq:pair-factorization}.  Its coefficient at
\(z_az_b\) is at least \(2u_\kappa^2\).  If
\[
  V_\ell=U_{r\ell}+U_{s\ell},
\]
then the coefficients at \(z_az_\ell\) and \(z_bz_\ell\) are at least
\(u_\kappa V_\ell\).  Also,
\[
  U_{i\ell}
  =\frac{X_{i\ell}^{1+\tau}}
         {\prod_{k\ne\ell}(1-X_{ik})}
  \ge X_{i\ell}^{1+\tau}.
\]
Convexity of \(x\mapsto x^{1+\tau}\) therefore gives
\begin{equation}\label{eq:V-lower}
  V_\ell
  \ge X_{r\ell}^{1+\tau}+X_{s\ell}^{1+\tau}
  \ge2^{-\tau}\alpha_\ell^{1+\tau}.
\end{equation}

We apply \cref{lem:capacity-certificate} to \(\widetilde Q\) with the feasible
distribution \(\theta\).  Its expected log coefficient is at least
\begin{equation}\label{eq:theta-coefficient}
  (1-\rho)\log(2u_\kappa^2)
  +\rho\log u_\kappa
  +\sum_{\ell\notin\set{a,b}}\alpha_\ell\log V_\ell.
\end{equation}
Its entropy is
\begin{align}
  H(\theta)
  &=
  -(1-\rho)\log(1-\rho)
  -\sum_{\ell\notin\set{a,b}}\alpha_\ell\log\alpha_\ell\notag\\
  &\quad
  -\delta_a\log\frac{\delta_a}{\rho}
  -\delta_b\log\frac{\delta_b}{\rho}.
  \label{eq:theta-entropy}
\end{align}
Using \cref{eq:V-lower} in \cref{eq:theta-coefficient}, and then adding
\cref{eq:theta-entropy}, gives
\begin{align}
  \log\capop_\alpha(\widetilde Q)
  &\ge
  (1-\rho)\log\frac{2u_\kappa^2}{1-\rho}
  +\rho\log u_\kappa
  -\tau\rho\log2\notag\\
  &\quad
  +\tau\sum_{\ell\notin\set{a,b}}
      \alpha_\ell\log\alpha_\ell
  -\delta_a\log\frac{\delta_a}{\rho}
  -\delta_b\log\frac{\delta_b}{\rho}.
  \label{eq:capacity-theta-bound}
\end{align}

\medskip\noindent\emph{Combining the factors.}
Insert the capacity bound into \cref{eq:pair-factorization}.  Its core
coordinates contribute
\(\delta_a\log\delta_a+\delta_b\log\delta_b\).  These terms combine
with the last line of \cref{eq:capacity-theta-bound} to give
\(\rho\log\rho\).  The outside coordinates satisfy
\[
  (1-\alpha_\ell)\log(1-\alpha_\ell)\ge-\alpha_\ell.
\]
Finally, \(\zeta_r\zeta_s\le1\).  We obtain
\begin{align}
  \log\Gamma_{rs}
  &\ge
  (1-\rho)\log\frac{2u_\kappa^2}{1-\rho}
  +\rho\log u_\kappa+\rho\log\rho-\rho\notag\\
  &\quad
  -\tau\left[
    -\sum_{\ell\notin\set{a,b}}
      \alpha_\ell\log\alpha_\ell
    +\rho\log2
  \right].
  \label{eq:gain-bound}
\end{align}

The entropy in brackets is at most
\(\rho\log(n/\rho)+\rho\log2\).  By
\cref{eq:tau-definition}, \(\log n\le\ell_n\log2\), and
\(\rho\log(1/\rho)\le1/e\), its contribution after multiplication by
\(\tau\) is at most \(\xi\) times
\[
  \frac14\left(
    \rho\log2+\frac1{e\ell_n}
    +\frac{\rho\log2}{\ell_n}
  \right).
\]
For all sufficiently small \(\kappa\), we have
\(\rho\le\overline\rho_\kappa<1/10\), so
this ratio is less than one.

For the first line of \cref{eq:gain-bound}, set
\[
  f_\kappa(\rho)
  =(1-\rho)\log\frac{2e^{-2\kappa}}{1-\rho}
   -\kappa\rho+\rho\log\rho-\rho.
\]
Here \(0\log0=0\).  Since
\(\overline\rho_\kappa=2(e^\kappa-1)\to0\),
\[
  \min_{0\le\rho\le\overline\rho_\kappa}f_\kappa(\rho)
  \longrightarrow\log2
  \qquad\text{as }\kappa\downarrow0.
\]
Choose a rational \(\kappa_0>0\) small enough that this minimum is larger
than \((\log2)/2\), and then choose rational \(\xi_0,\gamma_0>0\) such that
\[
  \gamma_0<\frac12\log2-\xi_0.
\]
Taking \(\kappa=\kappa_0\) in the preceding argument,
\cref{eq:gain-bound} proves the lemma for these constants.
\end{proof}

The same witnesses make the gain algorithmic without requiring a separate
capacity optimization.  For distinct columns \(a,b\) whose outside mass
\(\rho\) is at most one, define \(\theta^{a,b}\) by
\cref{eq:theta-distribution}, with the point mass on \(\set{a,b}\) when
\(\rho=0\), and write
\[
  c_{jk}=U_{rj}U_{sk}+U_{rk}U_{sj}.
\]
Set
\begin{align}
  \log\underline\Gamma_{rs}^{a,b}
  &:=-\log\zeta_r-\log\zeta_s
    +\sum_j(1-\alpha_j)\log(1-\alpha_j)\notag\\
  &\quad+\sum_{j<k}\theta^{a,b}_{jk}
       \log\frac{c_{jk}}{\theta^{a,b}_{jk}}.
  \label{eq:witness-gain}
\end{align}
By \cref{lem:capacity-certificate,lem:pair-factorization},
\(\underline\Gamma_{rs}^{a,b}\le\Gamma_{rs}\).  More importantly, the proof
of \cref{lem:clean-gain} used precisely the witness with the two core columns;
it therefore proves the stronger conclusion
\begin{equation}\label{eq:clean-witness-gain}
  \log\underline\Gamma_{rs}^{a,b}\ge\gamma_0
\end{equation}
for that choice of \(a,b\).

For the algorithm, we assign row pair \(\set{r,s}\) the nonnegative weight
\begin{equation}\label{eq:pair-weight}
  \underline w_{rs}
  =\max\set*{0,\max_{a\ne b:\ \rho_{rs}^{a,b}\le1}
    \log\underline\Gamma_{rs}^{a,b}}.
\end{equation}
If there is no eligible column pair, the inner maximum is omitted and the
weight is zero.  The core columns of every clean pair are eligible by
\cref{eq:alpha-leakage} and the choice of \(\kappa_0\).
A maximum-weight matching in the complete graph on the rows gives at least as
much total gain as any collection of clean pairs constructed only for the
analysis, and can be found in polynomial time
\cite{Schrijver2003}.  Edges of zero weight are omitted from the final
certificate.

\section{Completing the approximation bound}\label{sec:completion}

We now combine the stability information from the Bethe upper bound with the
gain from the paired lower certificate.  Let
\(\kappa_0,\gamma_0,\xi_0\) be the constants from
\cref{lem:clean-gain}.  Enlarge \(C_{\mathrm{stab}}\), if necessary, to make
it rational.

The exact values of the constants will not matter.  We choose them once, in
the order used by the proof: \(\eta\) controls the structural errors,
\(\delta_0\) defines the near-equality regime, and \(\xi\) controls the
regularization loss.  Choose rational \(\eta,\delta_0,\xi>0\), in this order,
and put
\[
  d_0=\left(\frac{\eta}{C_{\mathrm{stab}}}\right)^4.
\]
The following three quantities will bound, respectively, the fractions of
bad rows, rows on long components, and clean pairs with excessive transfer
cost.  We require \(\eta\le1/10\) and
\begin{align}
  \varepsilon_{\mathrm{row}}
  &:=\frac{\delta_0}{d_0}\le\frac1{128},
  \label{eq:row-error-choice}\\
  \varepsilon_{\mathrm{cyc}}
  &:=\frac6{\log2}\left[
    \delta_0
    +\left(1+\frac{\log2}{2}\right)\varepsilon_{\mathrm{row}}
    +\omega(\eta)
  \right]\le\frac1{16},
  \label{eq:cycle-error-choice}\\
  \varepsilon_{\mathrm{tr}}
  &:=\frac{
    \delta_0+2\xi+\Htwo(\eta)+\eta
    +(1+\log2)\varepsilon_{\mathrm{row}}
  }{(1/2-\eta)\kappa_0}\le\frac1{16},
  \label{eq:transfer-error-choice}\\
  \xi&\le\min\set*{\xi_0,\frac{\delta_0}{2},\frac{\gamma_0}{16}}.
  \label{eq:xi-choice}
\end{align}
Such a choice exists.  First take \(\eta\) small enough to control
\(\omega(\eta)\) and \(\Htwo(\eta)+\eta\); next take \(\delta_0\) small
relative to \(d_0\); finally take \(\xi\) small.  No numerical optimization
of these constants is needed.

Let \(X\) be the regularized point from \cref{sec:transfer}.  Give every row
pair the weight in \cref{eq:pair-weight}, and let \(M\) be a maximum-weight
matching.  Delete the zero-weight edges from \(M\), and define
\begin{equation}\label{eq:positive-certificate}
  L_X(A)
  =\exp\left(\beta_A(X)
   +\sum_{\set{r,s}\in M}\underline w_{rs}\right).
\end{equation}
For every positive-weight edge, choose columns attaining the maximum in
\cref{eq:pair-weight}.  Since its witness gain is no larger than
\(\Gamma_{rs}\), \cref{thm:paired-certificate} gives
\(L_X(A)\le\per(A)\).

\begin{proposition}[Positive matrices]\label{prop:positive-matrices}
There is an absolute constant \(\eps_+>0\) such that every positive
matrix \(A\) of order \(n\ge2\) satisfies
\begin{equation}\label{eq:positive-approximation}
  L_X(A)\le\per(A)
  \le\left(\sqrt2e^{-\eps_+}\right)^nL_X(A).
\end{equation}
\end{proposition}

\begin{proof}
The proof is a dichotomy between slack in the Bethe bound and gain from the
paired certificate.  Recall that \(b\) counts bad rows, \(N\) counts good rows
in long components of \(K\), \(D\) is the averaged sequential divergence, and
\(R_{\mathrm{core}}\) is the transfer cost on the edges of \(K\).

\paragraph*{Case 1: nonnegligible Bethe slack.}
Suppose that \(\Delta(A)\ge\delta_0n\).  Since the matching gain in
\cref{eq:positive-certificate} is nonnegative, \cref{eq:near-bethe} gives
\begin{align}
  \log\frac{\per(A)}{L_X(A)}
  &\le\log\per(A)-\beta_A(X)\notag\\
  &\le
  \left(\frac12\log2-\delta_0+\xi\right)n.
  \label{eq:far-case}
\end{align}

\paragraph*{Case 2: nearly tight Bethe bound.}
Suppose that \(\Delta(A)<\delta_0n\).  By
\cref{eq:bad-row-count,eq:row-error-choice},
\begin{equation}\label{eq:few-bad}
  b\le\varepsilon_{\mathrm{row}}n.
\end{equation}
Also \(D\le\Delta(A)\).  Rearranging \cref{eq:robust-cycle} and using
\cref{eq:cycle-error-choice} gives
\begin{equation}\label{eq:few-long}
  N\le\varepsilon_{\mathrm{cyc}}n.
\end{equation}
Hence \cref{eq:clean-cycle-count} gives at least
\[
  \frac{1-2\varepsilon_{\mathrm{row}}
    -\varepsilon_{\mathrm{cyc}}}{2}n
  \ge\frac{59}{128}n
\]
vertex-disjoint clean pairs.

By \cref{eq:core-transfer-normalized,eq:transfer-error-choice},
\begin{equation}\label{eq:small-core-transfer}
  R_{\mathrm{core}}
  \le
  \varepsilon_{\mathrm{tr}}(1/2-\eta)\kappa_0 n.
\end{equation}
If a clean pair with core columns \(a,b\) fails
\cref{eq:small-local-cost}, its contribution to \(R_{\mathrm{core}}\) is at
least \((1/2-\eta)\kappa_0\), because each of its four core \(P\)-entries is
at least \(1/2-\eta\).  Therefore at most
\(\varepsilon_{\mathrm{tr}}n\le n/16\) clean pairs fail the local-cost
condition.  At least \(3n/8\) vertex-disjoint pairs satisfy the hypothesis of
\cref{lem:clean-gain}.  The maximum-weight matching therefore has total gain
at least
\begin{equation}\label{eq:matching-gain}
  \sum_{\set{r,s}\in M}\underline w_{rs}
  \ge\frac{3\gamma_0}{8}n.
\end{equation}
Using the upper half of the Bethe sandwich and \cref{eq:near-bethe},
\begin{equation}\label{eq:near-case}
  \log\frac{\per(A)}{L_X(A)}
  \le
  \left(\frac12\log2+\xi-\frac{3\gamma_0}{8}\right)n.
\end{equation}

Define
\begin{equation}\label{eq:epsilon-plus}
  \eps_+
  =\min\set*{
    \delta_0-\xi,
    \frac{3\gamma_0}{8}-\xi
  }>0.
\end{equation}
\Cref{eq:far-case,eq:near-case} prove
\cref{eq:positive-approximation}.
\end{proof}

The proposition is the conceptual core of the theorem.  The remaining issues
are positivity and finite-precision computation; we handle them in the
appendix.

\appendix

\section{Zero entries and finite-precision computation}\label{sec:zeros}

This appendix passes from the exact positive-matrix argument to the algorithm
in \cref{thm:main}.  There are two issues: zero entries, which we remove by a
one-sided smoothing, and the finite-precision computation of the regularized
optimizer and the paired gains.  Neither issue changes the structural
argument.

\subsection{Smoothing zero entries}\label{subsec:smoothing}

We first extend the positive-matrix result to nonnegative matrices.  We choose
the smoothing to have polynomial bit complexity and to preserve a one-sided
certificate.

\begin{lemma}[Smoothing]\label{lem:smoothing}
Suppose that a nonnegative rational matrix \(A\) has largest entry at most one
and the bipartite graph of its positive entries has a perfect matching.  There
is a polynomial-time construction of a positive rational matrix \(\widetilde A\)
such that
\begin{equation}\label{eq:smoothing-comparison}
  \per(A)
  \le\per(\widetilde A)
  \le\left(1+\frac{\chi n}{2}\right)\per(A),
\end{equation}
where \(\chi>0\) is any fixed rational constant.
\end{lemma}

\begin{proof}
Let \(m\) be the smallest positive entry of \(A\), and put
\begin{equation}\label{eq:smoothing-level}
  \delta_{\mathrm{sm}}
  =\min\set*{
    \frac1{2n},
    \frac{\chi m^n}{4n!}
  },
  \qquad
  \widetilde A=A+\delta_{\mathrm{sm}}J.
\end{equation}
Here \(J\) is the all-ones matrix.
For every permutation \(\sigma\),
\[
  0\le
  \prod_i(A_{i,\sigma(i)}+\delta_{\mathrm{sm}})
  -\prod_iA_{i,\sigma(i)}
  \le(1+\delta_{\mathrm{sm}})^n-1
  \le2n\delta_{\mathrm{sm}}.
\]
For the last inequality, use
\((1+\delta_{\mathrm{sm}})^n
\le\sum_{k\ge0}(n\delta_{\mathrm{sm}})^k
=(1-n\delta_{\mathrm{sm}})^{-1}\) and
\(n\delta_{\mathrm{sm}}\le1/2\).
Summing over the \(n!\) permutations gives
\[
  0\le\per(\widetilde A)-\per(A)
  \le2nn!\delta_{\mathrm{sm}}
  \le\frac{\chi n}{2}m^n.
\]
The positive support contains a perfect matching, so
\(\per(A)\ge m^n\).  This proves \cref{eq:smoothing-comparison}.  The bit
length of \(\delta_{\mathrm{sm}}\) is polynomial in the bit length of \(A\)
and in \(n\).
\end{proof}

\subsection{Certified finite-precision computation}
\label{subsec:finite-precision}

It remains to compute the regularized point and the explicit witness gains
with certified precision while preserving the lower-bound direction.  We
first bound the optimizer away from the boundary.  We then solve a bounded
epigraph problem by weak optimization, recover approximate KKT potentials,
and interpret the result as an exact optimizer for a nearby matrix.  This last
step lets us reuse the exact certificate proved in the main argument.

\begin{lemma}[Interior bound]\label{lem:interior-bound}
Assume \(n\ge2\) and \(m\le A_{ij}\le1\) for every \(i,j\), where \(m>0\),
and let \(X\) be the maximizer of \(\Phi_\tau\).  Then
\begin{equation}\label{eq:interior-bound}
  \log\frac1{\min_{i,j}X_{ij}}
  \le
  \frac{nR_A}{\tau}+n^2\log n,
\end{equation}
where one may take
\begin{equation}\label{eq:objective-range-bound}
  R_A=n\log\frac1m+n\log n+n.
\end{equation}
In particular, if \(A\) is rational, every \(X_{ij}\) and every
\(1-X_{ij}\) is at least \(2^{-\poly(n,B,1/\tau)}\), where \(B\) is the input
bit length.
\end{lemma}

\begin{proof}
Let \(J\) be the all-ones matrix, set \(\overline X=J/n\), and write
\[
  \Phi_\tau(Z)=\beta_A(Z)+\tau\sum_iH(Z_i).
\]
The range of \(\beta_A\) on \(\cB_n\) is at most \(R_A\): the linear term
\(\sum_{i,j}Z_{ij}\log A_{ij}\) lies between \(n\log m\) and zero, the row
entropy is at most \(n\log n\), and
\[
  -\sum_{i,j}Z_{ij}
  \le \sum_{i,j}(1-Z_{ij})\log(1-Z_{ij})\le0.
\]
Here the first inequality follows from
\(-x\le(1-x)\log(1-x)\), and \(\sum_{i,j}Z_{ij}=n\).

Since \(X\) is an interior maximizer, the directional derivative of
\(\Phi_\tau\) at \(X\) toward \(\overline X\) is zero.  Concavity of
\(\beta_A\) gives
\[
  \beta_A'(X;\overline X-X)
  \ge\beta_A(\overline X)-\beta_A(X)\ge-R_A.
\]
Consequently
\begin{equation}\label{eq:entropy-direction-upper}
  \left(\sum_iH(X_i)\right)'(X;\overline X-X)\le\frac{R_A}{\tau}.
\end{equation}
The entropy derivative is
\begin{align*}
  -\sum_{i,j}\left(\frac1n-X_{ij}\right)(1+\log X_{ij})
  &=
  \sum_{i,j}X_{ij}\log X_{ij}
  -\frac1n\sum_{i,j}\log X_{ij},
\end{align*}
because the direction has coordinate sum zero.  Let
\(x_\star=\min_{i,j}X_{ij}\).  We have
\[
  \sum_{i,j}X_{ij}\log X_{ij}\ge-n\log n,
  \qquad
  -\frac1n\sum_{i,j}\log X_{ij}
  \ge\frac1n\log\frac1{x_\star}.
\]
Substitute these estimates into \cref{eq:entropy-direction-upper} to obtain
\cref{eq:interior-bound}.  Finally,
\[
  1-X_{ij}=\sum_{k\ne j}X_{ik}\ge(n-1)x_\star,
\]
which gives the same polynomial logarithmic bound for \(1-X_{ij}\).
\end{proof}

\begin{lemma}[Certified computation]\label{lem:certified-computation}
There is an absolute constant \(C_{\mathrm{num}}\) with the following
property.  For every positive rational matrix \(A\) of order \(n\ge2\) and
every rational \(\varepsilon_{\mathrm{num}}>0\), a deterministic algorithm
computes a rational number \(\widehat L\) such that
\begin{equation}\label{eq:certified-approximation}
  \widehat L\le\per(A)
  \le
  \left(\sqrt2e^{-\eps_++C_{\mathrm{num}}\varepsilon_{\mathrm{num}}}\right)^n
  \widehat L.
\end{equation}
The running time is polynomial in the input bit length, \(n\), and
\(\log(2+1/\varepsilon_{\mathrm{num}})\).
\end{lemma}

\begin{proof}
Replacing \(\varepsilon_{\mathrm{num}}\) by
\(\min\set*{\varepsilon_{\mathrm{num}},1}\), we may assume it is at most
one.  Recall that \(\tau=\xi/(4\ell_n)\) is rational and that both its bit
length and \(1/\tau\) are polynomially bounded.

\medskip\noindent\emph{A compact optimization domain.}
After a global scaling, assume \(m\le A_{ij}\le1\), where \(m\) is the
smallest entry of the scaled matrix, and let \(B\) be its input bit length.  By
\cref{lem:interior-bound}, the optimizer belongs to the truncated Birkhoff
polytope
\[
  \cB_n(\delta)
  =\set*{Z\in\cB_n\given Z_{ij}\ge\delta\ \forall i,j}
\]
for an explicitly computable rational
\(\delta_{\mathrm{int}}=2^{-\poly(n,B,1/\tau)}\).  Concretely, increase
\(B\) if necessary so that \(m\ge2^{-B}\), and put
\[
  K_0=\frac{n(nB+2n^2)}{\tau}+n^3,
  \qquad q=\ceil{2K_0},
  \qquad \delta_{\mathrm{int}}=2^{-q}.
\]
The inequalities \(\log 2\le1\) and \(\log 2\ge1/2\) show directly from
\cref{eq:interior-bound} that this floor is valid.  We optimize over
\(\cB_n(\delta_{\mathrm{int}}/2)\), which still contains the exact optimizer.
Since \(q\ge n+1\),
\(\delta_{\mathrm{int}}\le1/(2n)\).
On the slightly larger set \(\cB_n(\delta_{\mathrm{int}}/4)\), both
\(Z_{ij}\) and \(1-Z_{ij}\) are bounded below by
\(\delta_{\mathrm{int}}/4\).

Write \(f=-\Phi_\tau\).  Its gradient is
\begin{equation}\label{eq:negative-regularized-gradient}
  \nabla f(Z)_{ij}
  =-\log A_{ij}+(1+\tau)\log Z_{ij}
    +\log(1-Z_{ij})+(2+\tau).
\end{equation}
On the larger truncated set, its Frobenius norm is at most
\[
  G=n\bigl[B+(2+\tau)(q+3)\bigr],
\]
where we used \(\log2\le1\).  Thus \(f\) is \(G\)-Lipschitz there.  Its
gradient is \(L\)-Lipschitz there, with
\begin{equation}\label{eq:gradient-lipschitz}
  L=\frac{4(2+\tau)}{\delta_{\mathrm{int}}}.
\end{equation}
Both \(G\) and \(\log L\) have polynomial size.

\medskip\noindent\emph{A bounded epigraph and its separation oracle.}
We eliminate the row- and column-sum equations by rational affine coordinates
and apply weak linear optimization to the epigraph of \(f\).  We record the
geometry needed for this reduction.  Use the upper-left
\((n-1)\)-by-\((n-1)\) block as coordinates and recover the last row and
column from the sum constraints.  The resulting affine map \(T\) has norm at
most \(n\) from Euclidean coordinate error to Frobenius error.  Also
\[
  \abs{f(Z)}\le M_0:=nB+(1+\tau)n^2+n
  \qquad (Z\in\cB_n).
\]
We use the bounded epigraph
\[
  K=\set*{(y,t)\given T(y)\in\cB_n(\delta_{\mathrm{int}}/2),
       \ f(T(y))\le t\le M_0+2}.
\]
Since \(\norm{y}_2\le n\) and \(-M_0\le t\le M_0+2\) on \(K\), it lies
in the ball centered at the origin of rational radius
\(R=M_0+n+2\).  If \(y_U\)
represents the uniform matrix, then \((y_U,M_0+1)\) is the center of a ball
in \(K\) of radius at least
\[
  r=\min\set*{\frac14,\frac1{4n^2},\frac1{4nG}}.
\]
Indeed, changing \(y_U\) by Euclidean distance at most \(r\) changes every
entry of \(T(y_U)\) by at most \(nr\le1/(4n)\), so all entries remain above
\(\delta_{\mathrm{int}}/2\).  Changing \(t\) by at most \(r\le1/4\)
preserves its upper bound.  Finally,
\(f(T(y))\le f(T(y_U))+Gnr\le M_0+1/4\), whereas
\(t\ge M_0+3/4\), so the lower epigraph inequality also holds.
Thus the outer radius and \(\log(1/r)\) have polynomial bit length, exactly as
required by the weak optimization theorem of
\cite{GroetschelLovaszSchrijver1981}.

The weak separation oracle first checks all rational linear inequalities
defining \(\cB_n(\delta_{\mathrm{int}}/2)\) and the upper bound on \(t\)
exactly.  At a point passing these checks, directed approximations to
\cref{eq:negative-regularized-gradient} and to \(f\) give the usual supporting
hyperplane to the epigraph, with normal
\((T_{\mathrm{lin}}^*\nabla f(T(y)),-1)\).  Convexity implies that its value
on the displacement from the query to any \((z,s)\in K\) satisfies
\[
  \left\langle
    (T_{\mathrm{lin}}^*\nabla f(T(y)),-1),(z-y,s-t)
  \right\rangle
  \le t-f(T(y)).
\]
If the sign of the violation \(f(T(y))-t\) cannot be certified at the
requested precision, the oracle declares the query to be in the prescribed
weak enlargement.
This is the weak-separation model used in the separation--optimization
reduction.  Taking the value error and the gradient error times the outer
radius below a fixed fraction of the requested weak tolerance makes each
reported cut valid.  Evaluation to \(p\) directed bits takes time polynomial
in \(n,B,q,p\), so this additional accuracy is polynomially bounded.

\medskip\noindent\emph{A rounded ellipsoid implementation.}
We use a fixed dyadic precision for the entire ellipsoid run.  This avoids
evaluating the determinant of an intermediate rational matrix.  To see that
such a precision exists, clear the denominators of the initial rational
basis.  If their least common multiple is \(D\), nonsingularity gives the
explicit lower bound \(D^{-d}\) on the absolute determinant, where \(d\) is
the dimension of the epigraph.  Choose \(L\) so that \(2^{-L}\le D^{-d}\),
and choose \(K\) so that the initial center and basis have total absolute
magnitude at most \(2^K\).  After \(t\) rounded central cuts, the determinant
is at least \(2^{-(L+2t)}\), while the total magnitude is at most
\(2^{K+t(6+3d)}\).  The first estimate follows because the exact central cut
retains at least half of the determinant and sufficiently fine rounding
retains another half.  The second follows from the entrywise central-cut and
rounding bounds.  It is therefore enough to use, at every iteration, the
single precision
\[
  P=(L+2T+1)+12+8d+d^2
    +d\bigl(K+T(6+3d)+3+d\bigr)+2,
\]
where \(T\) is the iteration budget.  The determinant perturbation bound and
the adjugate formula show that flooring the center and basis to the
\(2^{-P}\) grid, followed by inflation by
\(1+1/(1024d^4)\), preserves the surviving half-ellipsoid.  The same estimates
give determinant contraction by a factor at most
\(1-1/(32d^3)\).  Thus \(P\) has polynomial length, is computed before the
run from rational data, and controls every stored state.

\medskip\noindent\emph{Recovering approximate KKT potentials.}
Let \(\Delta_{\mathrm{obj}}>0\) be the desired objective error.  Invoke weak
optimization with accuracy
\begin{equation}\label{eq:weak-optimization-accuracy}
  \sigma\le
  \min\set*{
    \frac{\delta_{\mathrm{int}}}{4n},
    \frac{\Delta_{\mathrm{obj}}}{2+nG}
  }.
\end{equation}
Write the rational output as \((\widehat y,\widehat t)\).  It is within
\(\sigma\) of some \((y',t')\in K\), and \(\widehat t\) is within
\(\sigma\) above the minimum.  Mapping \(\widehat y\) through \(T\) gives a
rational matrix \(\widehat X\).  Its
row and column sums are exactly one.  The first bound in
\cref{eq:weak-optimization-accuracy} and the norm bound for \(T\) imply
\(\widehat X_{ij}\ge\delta_{\mathrm{int}}/4\); exact row sums then give the
same lower bound for \(1-\widehat X_{ij}\).  The second bound and Lipschitzness
give, more explicitly,
\[
  f(T(\widehat y))
  \le f(T(y'))+nG\sigma
  \le t'+nG\sigma
  \le \widehat t+(1+nG)\sigma
  \le \min_K t+(2+nG)\sigma.
\]
Consequently,
\begin{equation}\label{eq:weak-objective-error}
  \Phi_\tau(X)-\Phi_\tau(\widehat X)\le\Delta_{\mathrm{obj}}.
\end{equation}
Thus the fact that a weak optimizer need not return a point in its input body
does not compromise exact double stochasticity.

The function \(f\) is \(\tau\)-strongly convex on the affine hull of
\(\cB_n\), because the negative Shannon term has Hessian
\(\tau\diag(1/Z_{ij})\succeq\tau I\), while concavity of the Bethe objective
makes the remaining Hessian positive semidefinite on the tangent space.
Strong convexity gives
\[
  \norm{\widehat X-X}_F
  \le
  \sqrt{2(\Phi_\tau(X)-\Phi_\tau(\widehat X))/\tau}.
\]
Set
\(\Delta_{\mathrm{obj}}
=\tau\varepsilon_{\mathrm{num}}^2/(128L^2)\) in
\cref{eq:weak-optimization-accuracy}.  Then the two gradients are within
\(\varepsilon_{\mathrm{num}}/8\) in Frobenius norm.  Approximate
the gradient rationally to accuracy \(\varepsilon_{\mathrm{num}}/16\).  The
exact optimizer's gradient has the form \((R_i+C_j)_{ij}\) by the KKT
equations.  Fix a row \(i_0\) and a column \(j_0\), let \(\widetilde G\) be
the rational gradient approximation, and set
\[
  \widetilde R_i=\widetilde G_{i j_0},
  \qquad
  \widetilde C_j=\widetilde G_{i_0j}-\widetilde G_{i_0j_0}.
\]
For an exact row-plus-column matrix this reconstruction is an identity.  If
each entry of \(\widetilde G\) is within \(a\) of such a matrix, its residual
is at most \(4a\), by the triangle inequality on the corresponding
four-cycle.  Our preceding choices make the residual, including the error in
evaluating the gradient at \(\widehat X\), at most
\(\varepsilon_{\mathrm{num}}\).  Set
\(\lambda_i=-\widetilde R_i\) and
\(\nu_j=2+\tau-\widetilde C_j\).  These rational potentials satisfy
\begin{equation}\label{eq:approximate-kkt}
  \left|
    \log A_{ij}-\lambda_i-\nu_j
    -(1+\tau)\log\widehat X_{ij}
    -\log(1-\widehat X_{ij})
  \right|\le\varepsilon_{\mathrm{num}}.
\end{equation}

\medskip\noindent\emph{Exactification by a nearby matrix.}
Define the nearby positive matrix
\begin{equation}\label{eq:nearby-matrix}
  A'_{ij}
  =e^{\lambda_i+\nu_j}
   \widehat X_{ij}^{1+\tau}(1-\widehat X_{ij}).
\end{equation}
Then \(\widehat X\) satisfies the exact KKT equations for \(A'\), and concavity
shows that it is the exact regularized optimizer.  Moreover,
\begin{equation}\label{eq:nearby-comparison}
  e^{-\varepsilon_{\mathrm{num}}}A'
  \le A\le e^{\varepsilon_{\mathrm{num}}}A'
\end{equation}
entrywise.  Hence
\begin{equation}\label{eq:per-nearby-comparison}
  e^{-\varepsilon_{\mathrm{num}}n}\per(A')
  \le\per(A)
  \le e^{\varepsilon_{\mathrm{num}}n}\per(A').
\end{equation}

\medskip\noindent\emph{One-sided evaluation of the paired certificate.}
The gain does not require another convex program.  For every row pair and
every pair of distinct candidate core columns, compute the outside mass and
retain the candidate only if it is at most one.  Then form the rational
distribution \(\theta^{a,b}\) in \cref{eq:theta-distribution}; the case
\(\rho=0\) is detected exactly and uses the point mass on \(\set{a,b}\).
These are exactly the finitely many witness gains in
\cref{eq:witness-gain}.  There are only \(O(n^4)\) of them.  Every retained
\(\theta^{a,b}\) satisfies its mass and marginal constraints exactly, so
\cref{lem:capacity-certificate} preserves the lower-bound direction without
any feasibility recovery.

Logarithms and exponentials at rational arguments can be evaluated to \(p\)
directed bits in time polynomial in the input length and \(p\)
\cite{Brent1976}.  We use outward-rounded interval arithmetic for all
elementary expressions generated from the rational inputs.  There is no need
to materialize the generally irrational matrix \(A'\): inside
\(\beta_{A'}(\widehat X)\) we use
\[
  \log A'_{ij}=\lambda_i+\nu_j
    +(1+\tau)\log\widehat X_{ij}+\log(1-\widehat X_{ij}),
\]
and the transfer coefficients are evaluated in logarithmic form from
\[
  \log U_{ij}=(1+\tau)\log\widehat X_{ij}
    -\sum_{k\ne j}\log(1-\widehat X_{ik}).
\]
The quantities \(\rho,\delta_a,\delta_b\) and every nonzero witness mass are
rational.  Zero masses are recognized exactly and their entropy terms are set
to zero.  We evaluate \(\beta_{A'}(\widehat X)\) and every witness log-gain,
always taking the lower endpoint.  For each row pair, take the largest lower
endpoint over the candidate columns, truncate it at zero, and compute a
maximum-weight matching.
If every edge weight is underestimated by at most
\(\varepsilon_{\mathrm{num}}\), the computed matching loses at most
\(n\varepsilon_{\mathrm{num}}/2\) relative to the exact witness-weight
matching, because both matchings have at most \(n/2\) edges.  Evaluating each
of the polynomially many elementary terms to
\(\varepsilon_{\mathrm{num}}/\poly(n)\) therefore gives a rational value
\(\widehat L'\) such that
\[
  \widehat L'\le L_{\widehat X}(A'),
  \qquad
  \log\widehat L'
  \ge\log L_{\widehat X}(A')-C_0\varepsilon_{\mathrm{num}}n
\]
for an absolute constant \(C_0\).  All inequalities are one-sided, so
\(\widehat L'\le\per(A')\) by \cref{thm:paired-certificate} and
\cref{eq:witness-gain}.

\medskip\noindent\emph{Assembling the output.}
Finally compute a rational \(\widehat L\) satisfying
\[
  e^{-2\varepsilon_{\mathrm{num}}n}\widehat L'
  \le \widehat L
  \le e^{-\varepsilon_{\mathrm{num}}n}\widehat L'.
\]
Such a number is obtained by directed evaluation with polynomially many bits.
The lower bound in
\cref{eq:certified-approximation} follows from
\cref{eq:per-nearby-comparison}.  Since \(\widehat X\) is the exact
regularized optimizer for \(A'\), \cref{prop:positive-matrices} gives
\[
  \per(A')
  \le
  \left(\sqrt2e^{-\eps_+}\right)^nL_{\widehat X}(A').
\]
Combining this inequality with \cref{eq:per-nearby-comparison} and the loss in
\(\widehat L'\) proves the upper bound in
\cref{eq:certified-approximation} with
\(C_{\mathrm{num}}=C_0+3\).  The initial scaling can be undone using the
degree-\(n\) homogeneity of the permanent and the certificate.  All requested
precisions have polynomial bit length, and every rational intermediate and
the final output have polynomial bit length.  Rational arithmetic, weak
optimization, directed elementary-function evaluation, and weighted matching
therefore take polynomial time in the ordinary binary encoding of the input.
\end{proof}

\subsection{Proof of the main theorem}\label{subsec:main-proof}

\begin{proof}[Proof of \cref{thm:main}]
The case \(n=1\) is exact.  For \(n\ge2\), first test by bipartite matching
whether the positive support of \(A\) contains a perfect matching.  If not,
return zero.

Otherwise scale \(A\) by its largest entry.  The constant \(\eps_+\) in
\cref{eq:epsilon-plus} is positive and rational.  Decreasing it if necessary,
we may assume that \(\eps_+\le(\log 2)/2\); this only weakens
\cref{eq:positive-approximation} and ensures that the final approximation
base is larger than one.
Apply \cref{lem:smoothing} with \(\chi=\eps_+/2\), and apply
\cref{lem:certified-computation} to the smoothed matrix with
\begin{equation}\label{eq:numerical-accuracy-choice}
  \varepsilon_{\mathrm{num}}
  \le\frac{\eps_+}{4C_{\mathrm{num}}}.
\end{equation}
By \cref{lem:certified-computation}, the smoothed permanent is bounded by the
computed certificate with base
\[
  \sqrt2e^{-\eps_++C_{\mathrm{num}}\varepsilon_{\mathrm{num}}}.
\]

If \(\widetilde L\) is the certified lower bound for the smoothed permanent,
put
\[
  L=\frac{\widetilde L}{1+\chi n/2}.
\]
By \cref{eq:smoothing-comparison}, \(L\le\per(A)\).  On the logarithmic scale,
the division loses at most
\[
  \log\left(1+\frac{\chi n}{2}\right)\le\frac{\chi n}{2}.
\]
Our choices of \(\chi\) and \(\varepsilon_{\mathrm{num}}\) leave at least
\(\eps_+/2\) in the exponent.  Thus the final approximation base can be taken
to be
\begin{equation}\label{eq:final-base}
  c=\sqrt2e^{-\eps_+/2}<\sqrt2.
\end{equation}
Undoing the initial scaling multiplies both the output and the permanent by
the same factor to the \(n\)th power.  Every step is deterministic and
polynomial-time.
\end{proof}

\section{Provenance}\label{app:provenance}

The proof in this paper originated in an interaction between the author and
ChatGPT 5.6 Sol Pro.  The author supplied the high-level plan of attack and
guided the development of the argument.  Within this plan, the model proposed
the paired stable-polynomial certificate, the exact-slack strategy, and the
transfer through a regularized Bethe point.  The author subsequently verified
the results.

Codex assisted with alternating manuscript expansion and critical proof passes,
including symbolic and numerical checks, literature verification, assembly,
and typesetting.  The author is responsible for the final contents.

An accompanying Lean 4 development, available at
\url{https://github.com/nimaanari/formalization-beyond-bethe} and registered
with the Palomar Registry at
\url{https://palomar-registry.org/entry?id=PALOMAR-2026-08-29-000006&version=1},
formalizes the proof and the complete algorithm.  It reconstructs
from their source proofs Vontobel's concavity theorem, the sharp one-row
inequality of Anari and Rezaei, the upper Bethe bound, and the stable coefficient
inequality of Anari and Oveis Gharan.  The lower Bethe bound is recovered
internally as a special case.

The formalization includes the convex optimization, finite-precision error
analysis, support and boundary cases, and an executable rational algorithm.  It
also proves ordinary deterministic polynomial running time on a binary
encoding, including the elementary arithmetic, ellipsoid, matching, and
certificate routines used by the algorithm.  Its main theorem has no
project-specific hypotheses, and its axiom audit reports only the standard
logical axioms used by Lean and Mathlib.

\renewcommand*{\bibfont}{\small}
\PrintBibliography

\end{document}